\documentclass[10pt,journal,compsoc]{IEEEtran}
\def\BibTeX{{\rm B\kern-.05em{\sc i\kern-.025em b}\kern-.08em
		T\kern-.1667em\lower.7ex\hbox{E}\kern-.125emX}}

\usepackage{cleveref}
\usepackage{varioref}
\usepackage{amsmath,amsthm,amssymb}
\usepackage{bbding}
\usepackage[table]{xcolor}

\usepackage{etoolbox}
\newtheorem{theorem}{Theorem}
\AtEndEnvironment{theorem}{\null\hfill\qedsymbol}%
\newtheorem{lem}{Lemma}
\usepackage{lipsum}
\usepackage{multicol,lipsum}
\usepackage{amsfonts}
\usepackage{booktabs}
\usepackage{epsfig}
\usepackage{cite}
\usepackage{color,soul}
\usepackage{subfigure}
\usepackage{multirow}
\usepackage{rotating}
\usepackage{tabularx}
\usepackage{array}
\usepackage{graphicx,dblfloatfix}
\usepackage{blindtext}
\usepackage[normalem]{ulem} 

\AtBeginDocument{%
}
\begin{document}
\title{Single or Multiple Frames Content Delivery for Next-Generation Networks?}
\author{Mohammad~R.~Abedi, Nader~Mokari, Mohammad~R.~Javan, and Eduard~.~A.~Jorswieck
	\thanks{
		
		Mohammad~R.~Abedi and Nader~Mokari are with ECE Department, Tarbiat
		Modares University, Tehran, Iran.
		
		Mohammad~R.~Javan is with the Department of Electrical and
		Robotics Engineering, Shahrood University, Shahrood,
		Iran.
		
		E~.~A.~Jorswieck is with the Department of Systems and Computer
		Engineering, Dresden University of Technology (TUD), Germany.}}
\maketitle

\begin{abstract}
	This paper addresses the four enabling technologies, namely multi-user sparse code multiple access (SCMA), content caching, energy harvesting, and physical layer security for
	 proposing an energy and spectral efficient resource allocation algorithm for the access and backhaul links in heterogeneous cellular networks.  Although each of the above mentioned issues could be a topic of research, in a real situation, we would face a complicated scenario where they should be considered jointly, and hence, our target is to consider these technologies jointly in a unified
	framework. Moreover, we propose two novel content delivery scenarios: 1) single frame content delivery (SFCD), and 2) multiple frames content delivery (MFCD), where the time duration of serving user requests is divided into several frames. In the first scenario, the requested content by each user is served over one frame. However, in the second scenario, the requested content by each user can be delivered over several frames. We formulate the resource allocation for the proposed scenarios as optimization problems where our main aim is to maximize the energy efficiency of access links subject to the transmit power and rate constraints of access and backhaul links, caching and energy harvesting constraints, and SCMA codebook allocation limitations. Due to the practical limitations, we assume that the channel state information values between eavesdroppers and base stations are uncertain and design the network for the worst case scenario.  Since the corresponding optimization problems are mixed integer non-linear and nonconvex programming, NP-hard, and intractable, we propose an iterative algorithm based on the well-known alternate and successive convex approximation methods. In addition, the proposed algorithms are  studied from the computational complexity, convergence, and performance perspectives. Moreover, the proposed caching scheme outperforms the existing traditional caching schemes like random caching and most popular caching. We also study the effect of joint and disjoint considerations of enabling technologies for the performance of next-generation networks. We also show that the proposed caching strategy, MFCD and joint solutions have 43\%, 9.4\% and \%51.3 performance gain compared to no cahcing, SFCD and disjoint solutions, respectively.
\end{abstract}


\section{Introduction}
\subsection{Background and Motivation}
Over recent years, the growth of high data rate of mobile traffic, energy,
content storing, security, and limited knowledge of channels over
mobile networks are the major challenges of network design and
implementation. To tackle these issues and cope with the users' requirements, the next-generation of wireless communications is introduced which uses multiple advanced techniques such as energy harvesting (EH), physical layer (PHY)
security, new multiple access techniques, and content caching. Hence,
all of these issues must be considered together and efficient joint radio resource allocation and content placement algorithms must be applied to provide high performance for the designed networks. However,
devising efficient algorithms to handle
all these issues is a challenging task, and to the best our
knowledge, no research exists addressing all these issues
together in a unified framework. Although each of the mentioned issues could be an interesting research topic, our main contribution
is to study the joint effect of security, EH,
content caching, and imperfect and limited channel knowledge
in a unified joint access and backhaul links framework. In this regards,
we develop a comprehensive model and mathematical representation,
and design a robust resource allocation algorithm. Although the resulting
optimization problem is complicated, effective
optimization methods are used to achieve the solution. The outline of each issue,  applicable solutions, and
related works are explained in the sequel.
\subsubsection{Growth of High Data Rate Mobile Traffic} Incredible growth
 in high data rate mobile applications requires high capacity in radio access and backhaul wireless links. However, the centralized nature of
mobile network architectures can not provide enough capacity on
the wireless access and backhaul links to satisfy high demand for
rich multimedia content. Heterogeneous
network consisting of multiple low power radio access nodes and  the traditional macrocell nodes, is a promising solution to improve coverage and to
provide high capacity \cite{mokari2016limited}. 

\subsubsection{Content Caching}
Multimedia services can
be provided using recent advanced mobile communication
technologies by  new types of mobile devices such as smart phones and tablets. However, transferring the same content several times in a short period imposes capacity pressures on the network. To overcome this, content caching at the
network edge has recently been emerged as a promising technique in
next-generation networks. Caching in next-generation mobile networks also reduces 
the mobile traffic by eliminating
the redundant traffic of duplicate transmissions of the same content
from servers. The
deployment of content cashing relevant to evolved packet core and radio access network (RAN) are studied in
\cite{Woo13}. By
caching, contents can be closer to the end-users, and backhaul traffic can be
offloaded \cite{Bastug14,rezvani2017fairness} to the edge of the network. The authors in
\cite{ahlehagh2014video,shanmugam2013femtocaching} investigate
caching the contents in RAN with the aim to store contents closer to
users. The content caching in small-cell base stations is studied
in \cite{shanmugam2013femtocaching,golrezaei2013femtocaching}. In \cite{gregori2015joint}, the authors reduce both the load and energy consumption of the backhaul links by caching the most popular contents at small base stations (SBSs). In \cite{zhang2017cost}, the authors consider two-tier heterogeneous wireless
networks (HetNets) with hierarchical caching, where the most popular files are cached at SBSs while the less popular ones are cached at macro base stations (MBSs). The goal of \cite{zhang2017cost} is to maximize network capacity with respect to the file transmission rate requirements by optimizing the cache sizes for MBSs and SBSs.

\subsubsection{Energy Harvesting}
 The offer of high-rate services increases the energy consumption at receivers which degrades the battery life.
Therefore, the trade-off between high-rate requirement and long
battery life is required to achieve good performance. Energy
harvesting has emerged as a promising approach to provide
sustainable networks with the long-term sustainable operation of
power supplies. In EH communication networks, nodes
acquire energy from environmental energy sources including random
motion and mechanical vibrations, light, acoustic, airflow, heat,
RF radio waves \cite{Yeatman04,Paradiso05}. The design of novel
transmission policies due to highly random and unpredictable
nature of harvestable profile of the harvested energy is required. 

\subsubsection{New Multiple Access Techniques}
Sparse code multiple access (SCMA) with near optimal spectral
efficiency is a promising technique to improve capacity of
wireless radio access \cite{Hosein13}. This multiple-access
technique that is based on non-orthogonal codebook assignment
provides massive connectivity and improves spectral efficiency
\cite{Hosein13,Au14,moltafet2017comparison}. By performing an appropriate codebook
assignment, a subcarrier in SCMA networks can be shared among
multiple users. Joint
codebook assignment and power allocation for SCMA is studied in
\cite{li2016joint}. The codebook assignment and power allocation
is also investigated in \cite{li2016cost}. The authors formulate
energy-efficient transmission problem to maximize the network energy efficiency (EE)
subject to system constraints.

\subsubsection{Imperfect Channel State Information}
In most previous works, the authors assume perfect channel
state information (CSI) of all links for BSs. However, in
practice, knowing of perfect CSI in BSs requires a huge amount
of bandwidth for signalling through the feedback links which is not
possible. Moreover, due to time varying channel, feedback delay,
quantization error, and estimation errors, perfect CSI may not be
available at transmitters. In this regard, some works aim to
tackle the performance degradation caused by the limited and
imperfect CSI
\cite{ahmed2006outage,mokari2016limited,javan2017resource}. In
\cite{ahmed2006outage}, the authors investigate the power and
subcarrier allocation by the quantized CSI. It is assumed that the
perfect CSI does not exist at transmitters and imperfect CSI can
be achieved via limited rate feedback channels. In
\cite{mokari2016limited}, joint power and subcarrier
allocation is studied for the uplink of an orthogonal
frequency-division multiple access (OFDMA) HetNet assuming imperfect CSI. In \cite{javan2017resource}, a
limited rate feedback scheme is considered to maximize the average
achievable rate for decode-and-forward relay cooperative
networks.

\subsubsection{Security}
The broadcast nature of wireless transmission makes security
against eavesdropping a major challenge for the next
generation wireless networks \cite{csiszar2013secrecy}. In this
regards, physical-layer security is a promising method to provide
security in wireless networks
\cite{wyner1975wire,alavi2017limited}. This technique explores the
characteristics of the wireless channel to provide security for
wireless transmission. In \cite{abedi2016limited}, the authors
consider physical layer security for relay assisted
networks with multiple eavesdroppers. They maximize the sum
secrecy rate of network with respect to transmission power
constraint for each transmitter via imperfect CSI. In
\cite{wang2017new}, the authors investigate the benefits of three
promising technologies, i.e., physical layer security, content
caching, and EH in heterogeneous wireless networks.

\subsubsection{Joint Backhaul and Access
Resource Allocation} Joint resource allocation at backhaul and access links is
investigated in \cite{sharma2017joint} for heterogeneous networks. In \cite{sharma2017joint}, the
full duplex self-backhauling capacity is used to simultaneously
communicate over the backhaul and access links. In \cite{dhifallah2015joint}, joint access
and backhaul links optimization is considered to minimize the
total network power consumption. In
\cite{hua2016wireless}, the authors study joint wireless backhaul
and the access links  resource allocation optimization. The goal is
to maximize the sum rate subject to the backhaul and access
constraints. Joint backhaul and access links optimization is
considered in \cite{shariat2015joint} for dense small cell
networks. Joint
resource allocation in access and backhaul links is considered for
ultra dense networks in \cite{zhuang2017joint} where the goal is to maximize the throughput of the network under system
constraints. In \cite{mirahsanjoint}, the authors consider joint
access and backhaul resource allocation for the
admission control of service requests in wireless virtual network. The access and backhaul links optimizations are considered
for small cells in the mmW frequency in \cite{niu2015exploiting}.

\subsection{Our Contributions}
This paper addresses the above joint provisioning of resources
between the wireless backhauls and access links by using
multi-user SCMA (MU-SCMA) to improve the network energy
efficiency. We consider secure communications in EH
enabled SCMA downlink communications with imperfect channel
knowledge. In our work, we combine and extend several techniques to improve
performance of network and formulate an optimization problem with the aim of maximizing EE with respect to system constraints. There are several
works which consider each of these topics separately. However, in
a real situation, these issues should be considered jointly. To the best
of our knowledge, none of the existing works considered the above issues in a unified framework.  The main contributions of this work are as follows:
\begin{itemize}
   \item We provide a unified framework in which physical layer security, content caching, EH, and imperfect knowledge of channel information is considered jointly in the design of wireless communication networks 
     \item We consider SCMA as a non-orthogonal multiple access technology where the codebooks are allowed to be used several times among users which increases the spectral efficiency.
     
       \item We propose two novel scenarios for content delivery, namely single frame content delivery (SFCD), and multiple frames content delivery (MFCD). We compare the performance of the proposed delivery scenarios with each other for different system parameters. Due to the random energy arrivals in the EH based communication, there may not be enough energy to send the entire file within the desired frames. Therefore, the first scenario may interrupt sending the file. To overcome this difficulty, we can use the second scenario. There, due to the file transfer in multiple frames, the probability of interrupting will be very low. It should be noted that the second scenario can be suitable for applications with large file sizes.
     
     \item We consider the access and backhaul links jointly and formulate the resource allocation for the proposed scenarios as optimization problems whose objectives are to maximize the energy efficiency of the network while transmit power and rate constraints, EH constraints, codebook assignment constraints, as well as caching constraints should be satisfied.
        
     \item We provide mathematical frameworks for our proposed resource allocation problems where fractional programming, alternative optimization, and successive convex approximation methods are successfully applied to achieve  solutions for the resource allocation optimization problems. We further study the convergence and the computational complexity of the proposed resource allocation algorithms.
     \item We evaluate and assess the performance of the proposed scheme for different values of the network parameters using numerical experiments. 
\end{itemize}

The following notations is used in the paper: $[x]^+=\max\{0,x\}$.
$|\mathcal{S}|$ denotes the cardinality of a set $\mathcal{S}$.
$[.]^{\dagger}$ represents the conjugate transpose. $\|.\|$ denotes the
Euclidean norm of a matrix/vector.

The rest of the paper is organized as follows. Section
\ref{System-Model} defines the system model. Section
\ref{THE OPTIMIZATION FRAMEWORK} is dedicated to the optimization
frameworks where the objectives and the constraints are explained. Section \ref{Resource Allocation based on worst case CSI}
describes the details of scheduling, power allocation
algorithm, content placement, EH, codebook assignment, and subcarrier allocation. In Section \ref{simulationsresults}, we provide the numerical analysis, and Section \ref{Conclusion} concludes the paper.

\section{System Model}\label{System-Model}
Consider the downlink SCMA transmission of a wireless heterogeneous
cellular network comprising of $O$ MBSs and $J$ SBSs in a two dimensional
Euclidean plane $\mathbb{R}^2$, as shown in Fig.
\ref{fig-System-Model}. Let us denote by
$\mathcal{O}=\{1,2,\dots,O\}$ the set of the MBSs and by
$\mathcal{J}=\{1,2,\dots,J\}$ the set of the SBSs. Each
cache-capable BS, i.e.,
$b\in\mathcal{B}=\{1,\dots,B\}=\mathcal{O}\bigcup\mathcal{J}$ with
size $B=|\mathcal{B}|$, is connected to the core network via backhaul\footnote{Point-to-multipoint (P2M) technologies are considered as backhaul
networks for small cell which is an effective way of sharing the
backhaul resource between several BSs. PMP backhaul has high
spectral efficiency, and speed and flexibility of deployment, and
have been successfully deployed in the Middle East, Africa and in
Europe by major operators \cite{BACKHAUL-10}.} 
links which are wireless links. The paper assumes that there is no interference
between the wireless backhaul and access links and these links are
out-of-band. A set of total number of users,
$\mathcal{U}_b=\{1,2,\dots,U_b\}$ is served by BS $b$
with size $U_b=|\mathcal{U}_b|$. The set of network users is
$\mathcal{U}=\bigcup_{b=1}^{B}\mathcal{U}_b$. The system consists
of $Q$ eavesdroppers which are indexed by
$q\in\mathcal{Q}=\{1,2,\dots,Q\}$ with size $Q=|\mathcal{Q}|$. The
total transmit bandwidth of the access, i.e., BW, is divided into $N$ subcarriers
where the bandwidth of each subcarrier
 is $BW_n$ ($BW=N\times
BW_n$). $K$ social media
$\omega_{k},k\in\mathcal{K}=\{1,\dots,K\}$, as the main traffic of
internet contents, are requested by the users in the network. We
assume that during the runtime of the network optimization
process,  user-BS association is fixed.

\begin{figure*}[h]
  \begin{center}
    \includegraphics[width=6 in]{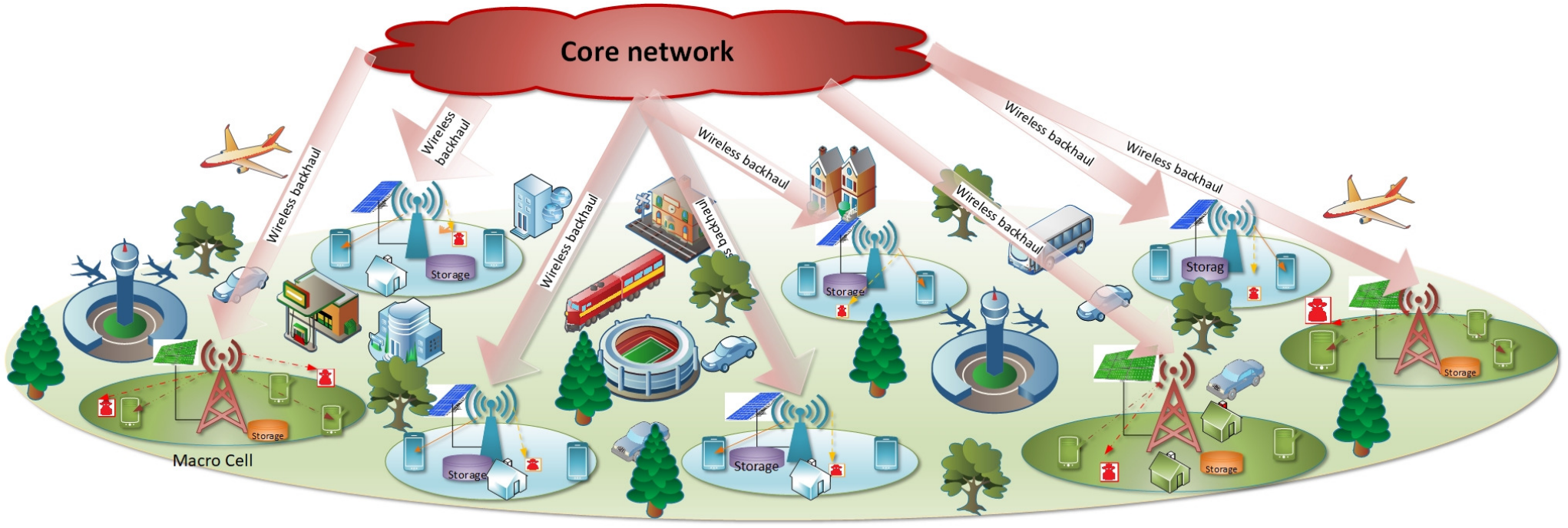} %
    \caption{The proposed wireless network with macro-BSs, small BSs, users, and eavesdroppers.}
    \label{fig-System-Model}
  \end{center}
\end{figure*}
\vspace{-0.5cm}
The message passing algorithm (MPA) can be used to detect
multiplexed signals on the same subcarriers \cite{Hoshyar08}. In our resource allocation framework, we consider two tasks: content caching and delivery resource allocations. The content caching task deals with determining which content should be cached in which storage. However, the delivery task deals with performing resource allocation such that the contents are delivered to the requesting users within serving time.  We assume that the time is split into several super frames. We further assume that each super frame is divided into  $F$ frames of duration $T$ seconds. throughout each supper frame, the arriving users requests, which should be served over the next supper frame, are gathered by the network control system. We emphasize that our proposed content caching and resource allocation algorithms are run for each super frame. Throughout the network run time, the network monitors the file requests and estimate the content distribution (content popularity). At the beginning of each super frame, if a change in the statistics of the contents popularity is detected, joint content caching and radio resource allocation is performed, and otherwise, only radio resource allocation is  performed. Note that the proposed resource allocation problem is solved at the beginning of each supper frame, and hence, the information about the CSIs and energy harvesting profile over all $F$ frames of the considered supper frame are required and should be known in advanced. With such assumption, we rely on the off-line approach which is common in the context of energy harvesting\footnote{In the context of energy harvesting, there is another approach which is called on-line approach. This approach assumes that the information is available only causally and use the Markov decision process method for resource allocation over $F$ frames \cite{minasian2014energy}. Although the availability of noncausal information is no practical, the off-line approach would provide a benchmark for energy harvesting networks. We leave the on-line approach as a future research direction} \cite{luo2013optimal,minasian2014energy}. The proposed transmission structure is shown in Fig. \ref{Frame-structure}.

\begin{figure*}[h]
	\vspace{-0.5cm}
  \begin{center}
    \includegraphics[width=6.4 in]{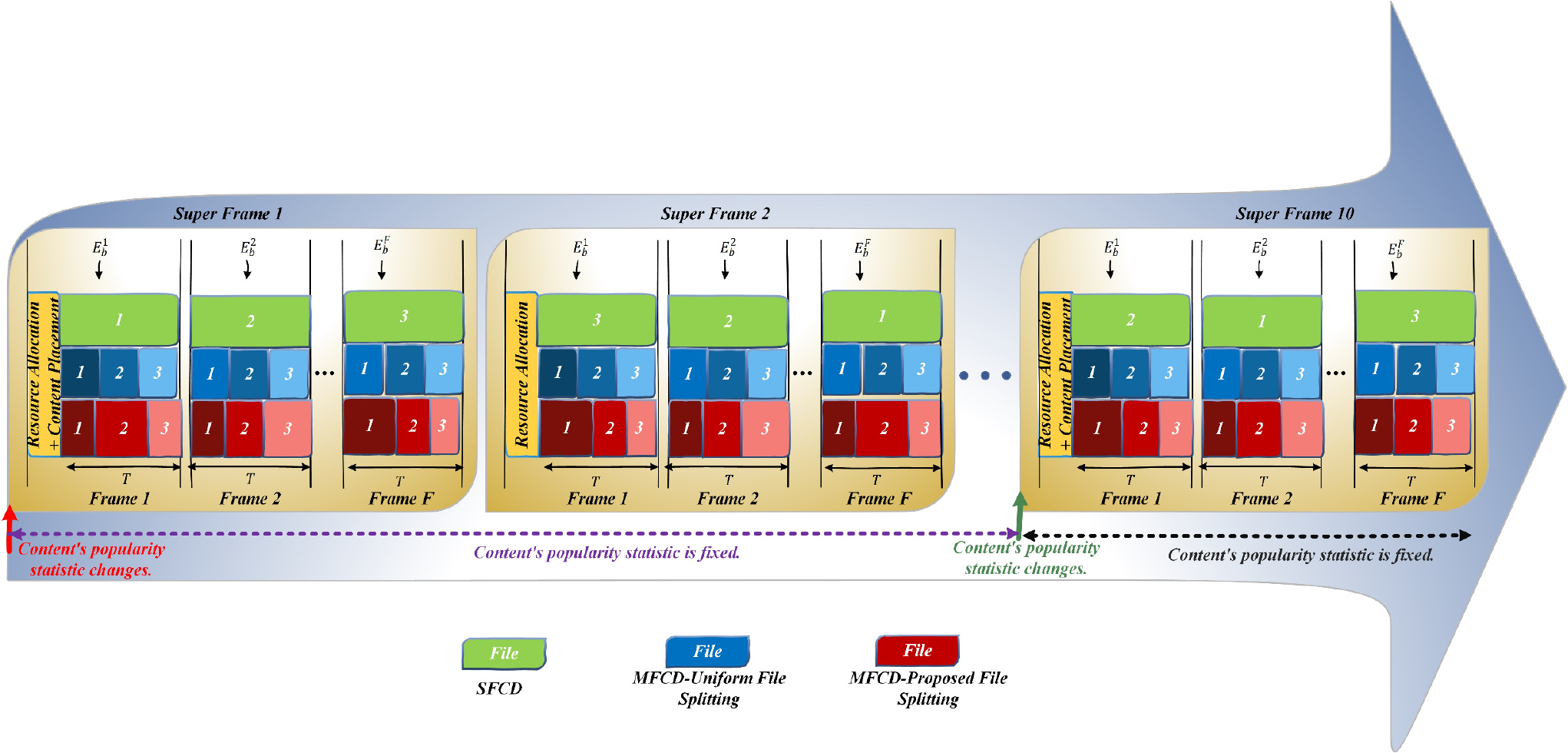} 
    \caption{SFCD and MFCD transmission structure.}
    \label{Frame-structure}
  \end{center}
\end{figure*}

Let $\mathbf{s}=\{s^{mt}_{bu}\}$ denote the codebook assignment at
BS $b$ at frame $t$ where $s^{mt}_{bu}$ is an indicator
variable that is 1 if codebook $m$ is assigned to user $u$ at BS
$b$ at frame $t$ and 0 otherwise. Furthermore, let
$\mathbf{p}=\{p^{mt}_{bu}\}$ denote the allocated transmit  power
vector with $p^{mt}_{bu}$ representing the transmit power for user
$u$ at BS $b$ at frame $t$ on codebook $m$. Thus, the total
transmit power of BS $b$ at frame $t$ is $\sum_{u\in
\mathcal{U}_b}\sum_{m\in
\mathcal{M}}s^{mt}_{bu}p^{mt}_{bu},\forall b\in \mathcal{B},t\in
\mathcal{F}$. To transmit the codewords to the designated users,
the transmit power $\mathbf{p}$ is finally allocated on the
corresponding subcarriers. However, different from OFDMA based
networks, the transmit power $p^{mt}_{bu}$ is allocated on
subcarrier $n$ according to a given proportion $\eta_{nm}$, which
is determined by the codebook design ($0<\eta_{nm}<1$ when
$c_{nm}=1$ and $\eta_{nm}=0$ when $c_{nm}=0$ \cite{Hosein13}).
Therefore, the signal-to-interference-plus-noise ratio (SINR) of
user $u$ in BS $b$ when using codebook $m$ can be expressed as follows:
\begin{equation}
\gamma^{mt}_{bu}=\frac{\sum_{n\in
\mathcal{N}}\eta_{nm}s^{mt}_{bu}p^{mt}_{bu}g^{nt}_{bu}}{I^{mt}_{bu}+(\sigma^{n}_u)^2},
\end{equation}
where $I^{mt}_{bu}=\sum_{\acute{b}\in
\mathcal{B}\setminus\{b\}}\sum_{\acute{u}\in
\mathcal{U}_{\acute{b}}}\sum_{n\in
\mathcal{N}}\eta_{nm}s^{mt}_{\acute{b}\acute{u}}
p^{mt}_{\acute{b}\acute{u}}g^{nt}_{\acute{b}u}$ and $g^{nt}_{bu}$ denotes the channel power gain between BS $b$ and
user $u$ on subcarrier $n$ at time $t$. $(\sigma^{n}_u)^2$ is the noise power on subcarrier $n$ at user $u$. Each of the subcarriers can be assumed to undergo a block-fading, and hence, the channel coefficients are kept constant within each frame. The achievable rate for the $u^{\text{th}}$ user in BS
$b$ at frame $t$ on codebook $m$ is given by $R^{\text{D},mt}_{bu}= \log_2\left(1+\gamma^{mt}_{bu}\right).$

We assume that the eavesdroppers only wiretap the access link\footnote{Due to the high computing power at the BSs, stronger cryptography is used for the links between core and BSs, hence eavesdropping of these links is hard and difficult. Therefore, we assume that only access link can be wiretapped by eavesdroppers.}. Therefore, the SINR of eavesdropper $q$ in BS $b$ when using
codebook $m$ can be expressed as:
\begin{equation}
\hat{\gamma}^{mt}_{buq}=\frac{\sum_{n\in
\mathcal{N}}\eta_{nm}s^{mt}_{bu}p^{mt}_{bu}
h^{nt}_{bq}}{\hat{I}^{mt}_{buq}+(\sigma^{n}_q)^2},
\end{equation}
where $\hat{I}^{mt}_{buq}=\sum_{\acute{b}\in
\mathcal{B}\setminus\{b\}}\sum_{\acute{u}\in
\mathcal{U}_{\acute{b}}}\sum_{n\in
\mathcal{N}}\eta_{nm}s^{mt}_{\acute{b}\acute{u}}p^{mt}_{\acute{b}\acute{u}}h^{nt}_{\acute{b}q}$
and $h^{mt}_{bq}$ denotes the channel power gain between BS $b$ and eavesdropper $q$ on subcarrier $n$. $(\sigma^{n}_q)^2$ is the noise power on subcarrier $n$ at eavesdropper $q$. The
achievable rate for the $q^{\text{th}}$ eavesdropper in BS $b$ at frame $t$ is evaluated by $R^{\text{E},mt}_{buq}=\log_2\left(1+\hat{\gamma}^{mt}_{buq}\right).$
The achievable secrecy access rate for non-colluding eavesdroppers and the $u^{\text{th}}$ user in BS $b$
at frame $t$ on codebook $m$ is expressed as \cite{pinto2012secure},
\begin{equation}\label{eq-R-S}
R^{\text{S},mt}_{bu}=\left[R^{\text{D},mt}_{bu}-\max_{q\in
\mathcal{Q}}R^{\text{E},mt}_{buq}\right]^+.
\end{equation}

\section{THE OPTIMIZATION FRAMEWORK}\label{THE OPTIMIZATION FRAMEWORK}
In this section, we provide the design objective and a
characterization of the constraints that must be satisfied by
content caching, EH, codebook assignment, and power
allocations. 

\subsection{System Constraints}
\subsubsection{Content Caching Constraints} Let the finite size of cache memory at
the $b^{\text{th}}$ BS is denoted by $V_b$. If the requested file
$k$ by user $u$ exists in the cache, then the file is sent to the user
immediately. This event is referred as a cache hit. However, if
file $k$ does not exist in the cache, then the request is
forwarded to the core network via backhaul, then downloaded file
$k$ from the core network via backhaul is forwarded to the user.
The size of the social media, $\alpha_k,k\in\mathcal{K}$ is
assumed to be Log-Normal distributed with parameters $\mu$ and $\kappa$ \cite{sobkowicz2013lognormal}. As the total cached media should not
exceed the finite size of cache memory at BS $b$, we have
\begin{equation}\label{eq-V}
\sum_{k\in\mathcal{K}}\theta_{bk}{\alpha}_k \leq V_b,\forall
b\in\mathcal{B},
\end{equation}
where $\theta_{bk}$ is a binary indicator declaring whether social media
$\omega_k$ is cached at BS $b$.

\subsubsection{
Content Delivery} The content delivery consists of two phases: 1) a cache placement phase, and 2)  a content delivery phase. In the cache placement phase, the cache content is determined at each BS, and in the content delivery phase, the requested files are delivered to users over wireless channels. In this paper, two new delivery scenarios are considered for content delivery phase. In the first scenario, the user's requested file $k$ with size $\alpha_k$ is sent in a single frame, while in the second scenario the user's requested file $k$ is divided into several parts with sizes $\{\beta^t_k\}, \forall t, k$, which  are sent over several frames. The scenarios are shown in Fig. \ref{Frame-structure}. To ensure that all parts of each file are transmitted to user, the following constraint should be satisfied
\begin{equation}\label{eq--1}
\sum_{t\in F}\beta^t_k=\alpha_k, \forall k.
\end{equation}

\subsubsection{Access and Backhaul Links Constraints}
Let $\upsilon_{ku}$ denote whether user $u$ needs $\omega_k$. The
backhaul traffic constraint for BS $b$ for the SFCD scenario is written as follows
\begin{align}\label{eq-Upsilon}
&\sum_{k\in\mathcal{K}}\sum_{u\in\mathcal{U}_b}\sum_{m\in\mathcal{M}}s^{mt}_{bu}
(1-\theta_{bk}).\min\left\{\sum_{u\in\mathcal{U}_b}\upsilon_{ku},1\right\}\alpha_k\\\nonumber&\leq
T\sum_{n\in\mathcal{N}}\zeta_{bn}\tilde{R}^{nt}_b,\forall
t\in\mathcal{F}, b\in\mathcal{B},
\end{align}
where the left hand side term of (\ref{eq-Upsilon}) is the backhual
traffic for BS $b$ and the right hand side term of (\ref{eq-Upsilon})
is backhaul traffic capacity, which must be greater than the backhaul traffic
for each BS. The backhaul link is a simple P2M link with OFDMA technology. $\zeta_{bn}\in\{0,1\}$ denotes whether BS $b$ uses subcarrier
$n$. For the MFCD scenario, $\alpha_k$ in (\ref{eq-Upsilon}) is replaced by $\beta^t_k$ as follows
\begin{align}\label{eq-Upsilon-1}
&\sum_{k\in\mathcal{K}}\sum_{u\in\mathcal{U}_b}\sum_{m\in\mathcal{M}}s^{mt}_{bu}
(1-\theta_{bk}).\min\left\{\sum_{u\in\mathcal{U}_b}
\upsilon_{ku},1\right\}\beta^t_k\leq\\\nonumber&
T\sum_{n\in\mathcal{N}}\zeta_{bn}
\tilde{R}^{nt}_b,\forall
t\in\mathcal{F}, b\in\mathcal{B}.
\end{align}

Note that for all requests of social media $\omega_k$ from BS $b$, if social media $\omega_k$ is not stored at BS $b$, the
requested social media $\omega_k$ is disseminated to BS $b$ from
the core network just once. $\tilde{R}^{nt}_b$ is the rate of
backhaul link for BS $b$ on subcarrier $n$ which is calculated by
\begin{equation}
\tilde{R}^{nt}_b=\log_2\left(1+\tilde{\gamma}^{nt}_b\right),\forall
t\in\mathcal{F},b\in\mathcal{B},n\in\mathcal{N}.
\end{equation}

When BS receives the data from core network, it transmits the file
back on the downlink. Let $\tilde{\gamma}^{nt}_b$ denotes the
received SNR at BS $b$ from the core network when the backhaul is
used to fetch the files from the core network for BS $b$.
$\tilde{\gamma}^{nt}_b$ can be written as $\tilde{\gamma}^{nt}_b=\frac{\tilde{p}^{nt}_b\tilde{h}^{nt}_b}{(\sigma^n_b)^2},$
where ${\tilde{p}^{nt}}_b$ is the transmit power of each wireless
backhaul link connected to BS $b$ on subcarrier $n$ and
$\tilde{h}^{nt}_b$ denotes the channel power gain between the $b^{\text{th}}$ BS on subcarrier $n$ and the core network and $(\sigma^n_b)^2$ is the
noise power at the $b^{\text{th}}$ BS on subcarrier $n$. Also the downlink traffic should not exceed the traffic capacity of each downlink. This yields for the first delivery scenario
\begin{equation}\label{eq-R-S-constraint}
\sum_{k\in\mathcal{K}}\sum_{m\in\mathcal{M}}s^{mt}_{bu}
\upsilon_{ku}\alpha_k\leq
T\sum_{m\in\mathcal{M}}R^{\text{S},mt}_{bu}, \forall
t\in\mathcal{F}, b\in\mathcal{B}, u\in\mathcal{U}_b.
\end{equation}

Note that for the MFCD scenario, $\alpha_k$ in (\ref{eq-R-S-constraint}) is replaced by $\beta^t_k$ as follows:
\begin{equation}\label{eq-R-S-constraint-1}
\sum_{k\in\mathcal{K}}\sum_{m\in\mathcal{M}}s^{mt}_{bu}
\upsilon_{ku}\beta^t_k\leq T
\sum_{m\in\mathcal{M}}R^{\text{S},mt}_{bu}, \forall
t\in\mathcal{F}, b\in\mathcal{B}, u\in\mathcal{U}_b.
\end{equation}

\subsubsection{Power Allocation Constraints} To determine the constraints
that must be satisfied by any feasible power allocation, let
$p^{mt}_{bu}$ and $\tilde{p}^{nt}_{b}$ denote the power allocated
to link the $b^{\text{th}}$ BS-the $u^{\text{th}}$ user at time
frame $t$ on codebook $m$ and to link core network-the
$b^{\text{th}}$ BS at frame $t$. The elements of $p^{mt}_{bu}$
and $\tilde{p}^{nt}_{b}$ must satisfy the followings:
\begin{equation}\label{eq-P-constraint-1}
p^{mt}_{bu}\geq0,\forall
b\in\mathcal{B},u\in\mathcal{U}_b,m\in\mathcal{M},t\in\mathcal{F},
\end{equation}
\begin{equation}\label{eq-tilde-P-constraint-1}
\tilde{p}^{nt}_{b}\geq0,\forall
b\in\mathcal{B},n\in\mathcal{N},t\in\mathcal{F}.
\end{equation}

In a practical network, core network has a power budget,
${P}^{\text{Total},t}$, which bounds the total power
allocated by core network on the core network-$b$ BS links and
subcarriers at frame $t$. This constraint can be written as:
\begin{equation}\label{eq-tilde-P-constraint-2}
\sum_{b\in\mathcal{B}}\sum_{n\in\mathcal{N}}\zeta_{bn}\tilde{p}^{nt}_{b}\leq
{P}^{\text{Total},t},\forall t\in\mathcal{F}.
\end{equation}

\subsubsection{EH Constraints}
We assume that the $b^{\text{th}}$ BS is connected to a
rechargeable battery with capacity $E^{\text{max}}_b$, and obtains
its power supply through an EH renewable sources such as solar. The renewable sources are used to charge batteries during the day. $E^t_b\in[0,E^{\text{max}}_b]$ is defined as the energy remaining
in the battery at the start of the $t^{\text{th}}$ frame. Then $E^t_b$ can be written in recursive form as:
\begin{align}\label{eq-energy-23456}\nonumber
E^{t+1}_b=&\min\left(E^t_b-T\sum_{m\in\mathcal{M}}
\sum_{u\in\mathcal{U}_b}s^{mt}_{bu}p^{mt}_{bu}+
\tilde{E}^t_b,E^{\text{max}}_b\right),\\& \forall b\in\mathcal{B}.
t\in\mathcal{F},
\end{align}
where $\tilde{E}^t_b$ denotes the amount of energy is harvested during the $t^{\text{th}}$ frame at the $b^{\text{th}}$ BS. The energy arrival takes place as a Poisson arrival process with mean $\Gamma_b$ \cite{ozel2011transmission,dhillon2014fundamentals}. The unit 
amount of energy harvested at each arrival at each BS is denoted by $\rho^t_b$, which depends
on the EH capabilities of renewable energy source at each BS. Therefore, $\tilde{E}^t_b=\varpi^t_b \rho^t_b$, where $\varpi^t_b$ is the number of
arrivals within $T$ with a mean value of $\Gamma_b T$. In
designing of optimal transmission policies for EH
communication systems, there are main constraints referred to as
energy consumption causality constraints, which state that the
energy packets which do not arrive yet, cannot be used by a
source. These constraints can be expressed as:
\begin{equation}\label{eq-energy-2}
\sum_{t=1}^{f}\sum_{m\in\mathcal{M}}\sum_{u\in\mathcal{U}_b}s^{mt}_{bu}p^{mt}_{bu}
\leq \frac{1}{T}\sum_{t=1}^{f}E^t_b, \forall b\in\mathcal{B},
f\in\mathcal{F}.
\end{equation}

If battery capacity is not enough to store the newly arrived
energy packet, the energy will be wasted at the beginning of a
transmission interval. By considering the following energy overflow constraint on
our problem, we avoid this battery overflow by enfrocing the following constraint:
\begin{equation}\label{eq-energy-3}
\sum_{t=1}^{f+1}E^t_b-T\sum_{t=1}^{f}\sum_{m\in\mathcal{M}}\sum_{u\in\mathcal{U}_b}s^{mt}_{bu}p^{mt}_{bu}
\leq E^{\text{max}}_b,\forall b\in\mathcal{B}, f\in\mathcal{F}.
\end{equation}

\subsubsection{Scheduling Constraints}
In order to improve the detection performance, we should use the codebooks which have less subcarriers in common. This means that, it must be guaranteed that each subcarrier cannot be reused
more than a certain value $D$, i.e., the maximum number of
differentiable constellations generated by the codebook-specific constellation function, as follows \cite{li2016cost}
\begin{equation}\label{eq-S}
\sum_{b\in\mathcal{B}}\sum_{u\in\mathcal{U}_b}\sum_{m\in\mathcal{M}}c_{nm}s^{mt}_{bu}\leq
D,\forall n\in\mathcal{N},t\in\mathcal{F}.
\end{equation}

In addition, (\ref{eq-S-1}), (\ref{eq-S-3}), and (\ref{eq-S-2}) together denote
that codebooks are exclusively allocated among users of each BS. For the SFCD scenario, we have
\begin{equation}\label{eq-S-1}
\sum_{m\in\mathcal{M}}
\sum_{t\in\mathcal{F}}
\sum_{u\in\mathcal{U}_b}
s^{mt}_{bu}\leq 1,\forall
b\in\mathcal{B},
\end{equation}
and for the MFCD scenario, we have
\begin{equation}\label{eq-S-3}
\sum_{m\in\mathcal{M}}
\sum_{u\in\mathcal{U}_b}s^{mt}_{bu}\leq 1,\forall b\in\mathcal{B},t\in\mathcal{F},
\end{equation}
\begin{equation}\label{eq-S-2}
s^{mt}_{bu}\in\{0,1\},\forall
b\in\mathcal{B},u\in\mathcal{U}_b,m\in\mathcal{M},t\in\mathcal{F}.
\end{equation}

\subsubsection{Worst Case Channel Uncertainty Model} For the
channels between the $b^{\text{th}}$ BS and the $q^{\text{th}}$
eavesdropper, only the estimated value $\tilde{h}^{nt}_{bq}$ is available at
the $b^{\text{th}}$ BS. We define the channel error as $e_{h^{nt}_{bq}}=|h^{nt}_{bq}-\tilde{h}^{nt}_{bq}|$, and we assume
that the channels mismatches are bounded as follows:
\begin{equation}\label{eq-e-3}
e_{h^{nt}_{bq}}\leq \varepsilon_{h^{nt}_{bq}},\forall b\in\mathcal{B},q\in\mathcal{Q},n\in\mathcal{N},t\in\mathcal{F},
\end{equation}
where $\varepsilon_{h^{nt}_{bq}}$ is
known constant. Hence the actual channel power gain value lies in the region $h^{nt}_{bq}\in\mathcal{H}^{nt}_{bq}=[\tilde{h}^{nt}_{bq}-\varepsilon_{h^{nt}_{bq}}~\tilde{h}^{nt}_{bq}+\varepsilon_{h^{nt}_{bq}}]$ \cite{liang2009compound}.

\subsection{The Optimization Problem}\label{The Optimization
Problem} We formulate the utility maximization problem
with power allocation, codebook assignment, and content caching
subject to energy causality and power budget constraints at each BS for the SFCD scenario as:
\begin{align}\label{eq--2}
\max_{\mathbf{p},\tilde{\mathbf{p}},\mathbf{s},\boldsymbol{\theta},\boldsymbol{\zeta}}
\min_{\textbf{h}\in\boldsymbol{\mathcal{H}}}&~
\Xi_{\text{EE}}(\mathbf{p},\mathbf{s}),\\\nonumber
\text{s.t.}~&(\ref{eq-V}),(\ref{eq-Upsilon}),(\ref{eq-R-S-constraint}),(\ref{eq-P-constraint-1})
-(\ref{eq-S-1}),(\ref{eq-S-2}),(\ref{eq-e-3}),
\end{align}
where $\Xi_{\text{EE}}(\mathbf{p},\mathbf{s})=\frac{\sum_{m\in
\mathcal{M}}\sum_{t\in \mathcal{F}}\sum_{b\in
\mathcal{B}}\sum_{u\in \mathcal{U}_b}R^{\text{S},mt}_{bu}}
{\sum_{m\in \mathcal{M}}\sum_{t\in \mathcal{F}}\sum_{b\in
\mathcal{B}}\sum_{u\in \mathcal{U}_b}s^{mt}_{bu}p^{mt}_{bu}}$, ${\textbf{h}}=[{\textbf{h}}^1,\dots,{\textbf{h}}^t,\dots,{\textbf{h}}^F],{\textbf{h}}^t=[h^{1t}_{11},\dots,h^{Nt}_{11},h^{1t}_{12},...,h^{Nt}_{1Q},...,h^{Nt}_{BQ}],\boldsymbol{\mathcal{H}}=\mathcal{H}^{11}_{11}\times\dots\times\mathcal{H}^{nt}_{bq}\times\dots\times\mathcal{H}^{NF}_{BQ}$. Note that for the MFCD scenario, constraint (\ref{eq--1}) is added to the optimization problem (\ref{eq--2}). We also replace (\ref{eq-Upsilon}), (\ref{eq-R-S-constraint}) and (\ref{eq-S-1}) by (\ref{eq-Upsilon-1}), (\ref{eq-R-S-constraint-1}) and (\ref{eq-S-3}), respectively. It should also be noted that in the second scenario, $\boldsymbol{\beta}$ is itself an optimization variable that must be obtained in the optimization problem. The optimization problem (\ref{eq--2}) consisting of non-convex objective function and both integer
and continuous variables. Hence, it is mixed-integer nonlinear programming (MINLP), non-convex, intractable and NP-hard problem \cite{murty1987some}.
\begin{lem}
	The optimization problem (\ref{eq--2}) is NP-hard.
\end{lem}
\begin{proof}
	Please see Appendix \ref{appendix-A}.
\end{proof}

It is very difficult to find the global
optimal solution within polynomial time. Hence, the available
methods to solve convex optimization problem can not be applied
directly. To solve this problem, an iterative algorithm based on the well-known and well-proven alternating, Dinkelbach and successive convex approximation methods is
proposed where in each iteration, the main problem is decoupled
into several sub-problems subject to some optimization variables.

\section{PROPOSED SOLUTION}\label{Resource Allocation based on worst case
CSI} The difficulty of solving the problem (\ref{eq--2}) arises from the nonconvexity of both the
objective function and feasible domain. As far as we know, there
is no standard method to solve such a nonconvex optimization
problem. In this section, some optimization methods such as
alternative optimization, fractional programming, and difference-of-two-concave-functions (DC) programming,
are jointly applied to solve the primal problem by transforming it
into simple subproblems step by step. To facilitate solving
(\ref{eq--2}), an alternate optimization method is
adopted to solve a multi-level hierarchical problem which consists of the several subproblem. The core idea of the
alternate optimization is that only one of the optimization parameters
is optimized in each step while others are fixed. When each
parameter is given, the resulting subproblem can be reformulated
as the form of DC problem and solved by DC programming. Moreover,
a sequential convex program is finally solved by convex
optimization methods at each iteration of the DC programming. In this section, we propose a solution for the SFCD scenario which is suitable for MFCD, too. The
transformation process for solving this problem mainly consists of
the following steps: I. \textit{Transformation of the primal problem:} By using the epigraph method, the inner maximization in the objection function in (\ref{eq-worstcase-p-123}) can be 
simplified and the secondary
    problem can be naturally derived. II. \textit{Alternate optimization over some variables:} In this step,
the alternate optimization method is adopted to cope with the
non-convexity of the resulting parametrized secondary problems
which is further rewritten as five sub-problems, namely, access power allocation, access code allocation, backhaul power allocation,  backhaul subcarrier allocation, content placement, and channel uncertainty. III.  \textit{DC programming for the nonconvex constraint
    elimination:} In this step, we reformulate the nonconvex constraint (\ref{eq-4321}) as a canonical DC programming which can be settled by iteratively solving a series
of sequential convex constraints. Finally, these convex constraints
can be solved by convex programming. IV. 
\textit{Fractional programming:} Applying fractional
programming, the parameterized secondary subproblem is solved with
a given parameter in each iteration.

\subsection{Transformation of the primal problem}
For simplifying (\ref{eq--2}), we herein introduce
auxiliary variables
$\boldsymbol{\varphi}=\{\varphi^{mt}_{bu}\in\mathbb{R}\}$.
Additionally, we can rewrite (\ref{eq--2}) 
equivalently as
\begin{subequations}\label{eq-worstcase-p-123}
\begin{align}
&\max_{\mathbf{p},\tilde{\mathbf{p}},\mathbf{s},\boldsymbol{\theta},\boldsymbol{\zeta},\boldsymbol{\varphi}}
\min_{\textbf{h}\in\boldsymbol{\mathcal{H}}}~A,\\\nonumber
&\text{s.t.}~\sum_{k\in\mathcal{K}}\sum_{m\in\mathcal{M}}s^{mt}_{bu}
\upsilon_{ku}\alpha_k\leq
\sum_{m\in
\mathcal{M}}\max\left\{R^{\text{D},mt}_{bu}-\varphi^{mt}_{bu},0\right\},\\&~~~~~~~~~~~~~~~~~~~~~~~~~~~~~~~~~~~~~~~~~~~~~~
\forall t\in\mathcal{F}, b\in\mathcal{B},
u\in\mathcal{U}_b,\\\label{eq-4321}&R^{\text{E},mt}_{buq}\leq\varphi^{mt}_{bu},\forall
m\in\mathcal{M},
t\in\mathcal{F},b\in\mathcal{B},u\in\mathcal{U}_b,q\in\mathcal{Q},
\\\nonumber&(\ref{eq-V}),(\ref{eq-Upsilon}),(\ref{eq-P-constraint-1}),(\ref{eq-tilde-P-constraint-1}),(\ref{eq-S})
-(\ref{eq-e-3}).
\end{align}
\end{subequations}
where $A=\frac{\sum_{m\in \mathcal{M}}\sum_{t\in
		F}\sum_{b\in \mathcal{B}}\sum_{u\in
		\mathcal{U}_b}\max\left\{R^{\text{D},mt}_{bu}-\varphi^{mt}_{bu},0\right\}}
{\sum_{t\in \mathcal{F}}\sum_{b\in \mathcal{B}}\sum_{m\in
		\mathcal{M}}\sum_{u\in
		\mathcal{U}_b}s^{mt}_{bu}p^{mt}_{bu}}.$
To solve the optimization problem (\ref{eq-worstcase-p-123}), we
should further transform it. We first rewrite
$\max\left\{R^{\text{D},mt}_{bu}-\varphi^{mt}_{bu},0\right\}$ as \cite{van2013solution}: $\max\left\{R^{\text{D},mt}_{bu}-\varphi^{mt}_{bu},0\right\}=
\max\left\{-R^{\text{D},mt}_{bu2}-\varphi^{mt}_{bu},-R^{\text{D},mt}_{bu1}\right\}
+R^{\text{D},mt}_{bu1}$ where
\begin{align}
&R^{\text{D},mt}_{bu1}=\log_2\Big(\sum_{b\in
\mathcal{B}}\sum_{u\in \mathcal{U}_b}\sum_{n\in \mathcal{N}}
\Big(\eta_{nm}s^{mt}_{bu}p^{mt}_{bu}g^{nt}_{bu}+(\sigma^{n}_u)^2\Big)\Big),
\\&R^{\text{D},mt}_{bu2}=\log_2\left(\sum_{\acute{b}\in
\mathcal{B}\setminus\{b\}}\sum_{\acute{u}\in
\mathcal{U}_{\acute{b}}}\sum_{n\in
\mathcal{N}}\left(\eta_{nm}s^{mt}_{\acute{b}\acute{u}}
p^{mt}_{\acute{b}\acute{u}}g^{nt}_{\acute{b}u}+(\sigma^{n}_u)^2\right)\right).
\end{align}

By introducing auxiliary variables
$\boldsymbol{\delta}=
\{\delta^{mt}_{bu}\in\mathbb{R}\}$,
(\ref{eq-worstcase-p-123}) is equivalently reformulated as \cite{van2013solution}
\begin{subequations}\label{eq-worstcase-p-1234}
\begin{align}
&\max_{\mathbf{p},\tilde{\mathbf{p}},\mathbf{s},\boldsymbol{\theta},\boldsymbol{\zeta},\boldsymbol{\varphi}
,\boldsymbol{\delta}}\min_{\textbf{h}\in\boldsymbol{\mathcal{H}}}
~\Theta(\mathbf{p},\tilde{\mathbf{p}},\mathbf{s},\boldsymbol{\theta},\boldsymbol{\zeta},\boldsymbol{\varphi},\mathbf{e}_h),\\\label{eq-EEE-1}
&\text{s.t.}~\sum_{k\in\mathcal{K}}\sum_{m\in\mathcal{M}}s^{mt}_{bu}
\upsilon_{ku}\alpha_k\leq
\sum_{m\in \mathcal{M}}\left\{\delta^{mt}_{bu}
+R^{\text{D},mt}_{bu1}\right\},\\\nonumber& \forall t\in\mathcal{F},
b\in\mathcal{B},
u\in\mathcal{U}_b,\\\label{eq-EEE}&R^{\text{E},mt}_{buq}\leq\varphi^{mt}_{bu},\forall
m\in\mathcal{M},
t\in\mathcal{F},b\in\mathcal{B},u\in\mathcal{U}_b,q\in\mathcal{Q},\\\label{eq-EEE-2}
&-R^{\text{D},mt}_{bu2}-\varphi^{mt}_{bu}
\leq\delta^{mt}_{bu},\forall m\in\mathcal{M},
t\in\mathcal{F},b\in\mathcal{B},u\in\mathcal{U}_b,
\\\label{eq-EEE-3}&-R^{\text{D},mt}_{bu1}\leq\delta^{mt}_{bu},\forall m\in\mathcal{M},
t\in\mathcal{F},b\in\mathcal{B},u\in\mathcal{U}_b,
\\\nonumber&(\ref{eq-V}),(\ref{eq-Upsilon}),(\ref{eq-P-constraint-1})
-(\ref{eq-e-3}).
\end{align}
\end{subequations}
where $\Theta(\mathbf{p},\tilde{\mathbf{p}},\mathbf{s},\boldsymbol{\theta},\boldsymbol{\zeta},\boldsymbol{\varphi},\mathbf{e}_h)=\frac{\sum_{m\in \mathcal{M}}\sum_{t\in
		F}\sum_{b\in \mathcal{B}}\sum_{u\in
		\mathcal{U}_b}\left\{\delta^{mt}_{bu}
	+R^{\text{D},mt}_{bu1}\right\}} {\sum_{t\in \mathcal{F}}\sum_{b\in
		\mathcal{B}}\sum_{m\in \mathcal{M}}\sum_{u\in
		\mathcal{U}_b}s^{mt}_{bu}p^{mt}_{bu}}$.

\subsection{Alternate optimization over optimization variables}
Due to the combined non-convexity of both objective function and
the constraint with respect to optimization parameters, the optimization 
problem (\ref{eq--2}) is difficult to solve. According
to alternate optimization method, we can always optimize a function by
first optimizing over some of the variables, and then optimizing
over the remaining ones. For convenience, the feasible domain of
(\ref{eq--2}) is denoted by $\mathbb{D}$ as $\mathbb{D}\triangleq
\left\{(\mathbf{p},\tilde{\mathbf{p}},\mathbf{s},\boldsymbol{\theta},\boldsymbol{\zeta},\mathbf{e}_h):
(\ref{eq-V}),(\ref{eq-Upsilon}),(\ref{eq-R-S-constraint})-(\ref{eq-e-3})\right\}$. For fixed
$\mathbf{p},\tilde{\mathbf{p}},\mathbf{s},\boldsymbol{\theta},\boldsymbol{\zeta}$,
$\mathbf{e}_h$-section of the feasible domain of $\mathbb{D}$, i.e.,
$\mathbb{D}_{\mathbf{e}_h}$, is defined as $\mathbb{D}_{\mathbf{e}_h}\triangleq \left\{\mathbf{e}_h:
(\mathbf{p},\tilde{\mathbf{p}},\mathbf{s},\boldsymbol{\theta},\boldsymbol{\zeta},\mathbf{e}_h)
\in\mathbb{D}\right\}$. Likewise, for fixed
$\tilde{\mathbf{p}},\mathbf{s},\boldsymbol{\theta},\boldsymbol{\zeta},\mathbf{e}_h$,
$\mathbf{p}$-section of the feasible domain of $\mathbb{D}$, i.e.,
$\mathbb{D}_{\mathbf{p}}$, is defined as 
$\mathbb{D}_{\mathbf{p}}\triangleq \left\{\mathbf{p}:
(\mathbf{p},\tilde{\mathbf{p}},\mathbf{s},\boldsymbol{\theta},\boldsymbol{\zeta},\mathbf{e}_h)
\in\mathbb{D}\right\}$. Similarly, for fixed
$\mathbf{p},\tilde{\mathbf{p}},\boldsymbol{\theta},\boldsymbol{\zeta},\mathbf{e}_h$,
$\mathbf{s}$-section of the feasible domain of $\mathbb{D}$, i.e.,
$\mathbb{D}_{\mathbf{s}}$, is defined as 
$\mathbb{D}_{\mathbf{s}}\triangleq \left\{\mathbf{s}:
(\mathbf{p},\tilde{\mathbf{p}},\mathbf{s},\boldsymbol{\theta},\boldsymbol{\zeta},\mathbf{e}_h)
\in\mathbb{D}\right\}$. In the same way, for fixed
$\mathbf{p},\mathbf{s},\boldsymbol{\theta},\mathbf{e}_h$,
$\tilde{\mathbf{p}}\times \boldsymbol{\zeta}$-section of the feasible domain of
$\mathbb{D}$, i.e., $\mathbb{D}_{\tilde{\mathbf{p}}\times \boldsymbol{\zeta}}$, is defined as 
$\mathbb{D}_{\tilde{\mathbf{p}}\times \boldsymbol{\zeta}}\triangleq
\left\{\tilde{\mathbf{p}},
\boldsymbol{\zeta}:
(\mathbf{p},\tilde{\mathbf{p}},\mathbf{s},\boldsymbol{\theta},
\boldsymbol{\zeta},\mathbf{e}_h)\in\mathbb{D}\right\}$. Correspondingly, for fixed
$\tilde{\mathbf{p}},\mathbf{s},\boldsymbol{\theta},\boldsymbol{\zeta},\mathbf{e}_h$,
$\boldsymbol{\theta}$-section of the feasible domain of
$\mathbb{D}$, i.e., $\mathbb{D}_{\boldsymbol{\theta}}$, is defined as 
$\mathbb{D}_{\boldsymbol{\theta}}\triangleq
\left\{\boldsymbol{\theta}:
(\mathbf{p},\tilde{\mathbf{p}},\mathbf{s},\boldsymbol{\theta},\boldsymbol{\zeta},\mathbf{e}_h)
\in\mathbb{D}\right\}$. Finally, the alternate optimization is used to solve the following
hierarchical five-level optimization subproblem:
	\begin{equation}
\max_{\substack{\mathbf{p}\in\mathbb{D}_{\mathbf{p}}\\\boldsymbol{\varphi}\in\mathbb{R}\\\boldsymbol{\delta}\in\mathbb{R}}}
\left[\max_{\mathbf{s}\in\mathbb{D}_{\mathbf{s}}}
\left[\max_{\substack{\tilde{\mathbf{p}},\boldsymbol{\zeta}\\\in\mathbb{D}_{\tilde{\mathbf{p}}\times \boldsymbol{\zeta}}}}
\left[\max_{\boldsymbol{\theta}\in\mathbb{D}_{\boldsymbol{\theta}}}
\left[\min_{\mathbf{e}_h\in\mathbb{D}_{\mathbf{e}_h}}
\Theta(\mathbf{p},\tilde{\mathbf{p}},\mathbf{s},\boldsymbol{\theta},\boldsymbol{\zeta},\boldsymbol{\varphi},\mathbf{e}_h)\right]\right]\right]\right].
\end{equation}

In conclusion, the subproblems can be solved sequentially at each
iteration of alternate optimization. In the first optimization
subproblem, we find $\mathbf{e}_h$ for a given
$\mathbf{p}_{\varrho},\mathbf{s}_{\varrho},\boldsymbol{\varphi}_{\varrho}$,
and $\boldsymbol{\delta}_{\varrho}$:
\begin{align}\label{AS-1}
\min_{\mathbf{e}_h\in\mathbb{D}_{\mathbf{e}_h}}
\Theta(\mathbf{p}_{\varrho},\tilde{\mathbf{p}}_{\varrho},\mathbf{s}_{\varrho},\boldsymbol{\theta}_{\varrho},\boldsymbol{\zeta}_{\varrho},\boldsymbol{\varphi}_{\varrho},{\mathbf{e}_h}_{\varrho}),
\end{align}
where $\varrho$ is the iteration number of alternate optimization
algorithm. By defining the solution of (\ref{AS-1}) as
${\mathbf{e}_h}_{\varrho+1}$, the second level subproblem is
solved to find $\boldsymbol{\theta}$ with a given
$\mathbf{p}_{\varrho},\mathbf{s}_{\varrho},\boldsymbol{\varphi}_{\varrho}$,
and $\boldsymbol{\delta}_{\varrho}$:
\begin{align}\label{AS-2}
\max_{\boldsymbol{\theta}\in\mathbb{D}_{\boldsymbol{\theta}}}
\Theta(\mathbf{p}_{\varrho},\tilde{\mathbf{p}}_{\varrho},\mathbf{s}_{\varrho},\boldsymbol{\theta}_{\varrho},\boldsymbol{\zeta}_{\varrho},\boldsymbol{\varphi}_{\varrho},{\mathbf{e}_h}_{\varrho}).
\end{align}
Similarly, by defining the solution of (\ref{AS-2}) as
$\boldsymbol{\theta}_{\varrho+1}$, the third level subproblem is
solved to find $\boldsymbol{\zeta}$ and $\tilde{\mathbf{p}}$ with a given
$\mathbf{p}_{\varrho},\mathbf{s}_{\varrho},\boldsymbol{\varphi}_{\varrho}$,
and $\boldsymbol{\delta}_{\varrho}$:
\begin{align}\label{AS-3}
\max_{\tilde{\mathbf{p}},\boldsymbol{\zeta}\in\mathbb{D}_{\tilde{\mathbf{p}}\times \boldsymbol{\zeta}}}
\Theta(\mathbf{p}_{\varrho},\tilde{\mathbf{p}}_{\varrho},\mathbf{s}_{\varrho},\boldsymbol{\theta}_{\varrho},\boldsymbol{\zeta}_{\varrho},\boldsymbol{\varphi}_{\varrho},{\mathbf{e}_h}_{\varrho}).
\end{align}

 Correspondingly, by defining the solution of (\ref{AS-3}) as
$\tilde{\mathbf{p}}_{\varrho+1}$ and $\boldsymbol{\zeta}_{\varrho+1}$, the fifth level subproblem is
solved to find $\mathbf{s}$ with a given
$\mathbf{p}_{\varrho},\boldsymbol{\varphi}_{\varrho}$, and
$\boldsymbol{\delta}_{\varrho}$:
\begin{align}\label{AS-4}
\max_{\mathbf{s}\in\mathbb{D}_{\mathbf{s}}}
\Theta(\mathbf{p}_{\varrho},\tilde{\mathbf{p}}_{\varrho},\mathbf{s}_{\varrho},\boldsymbol{\theta}_{\varrho},\boldsymbol{\zeta}_{\varrho},\boldsymbol{\varphi}_{\varrho},{\mathbf{e}_h}_{\varrho}).
\end{align}
Finally, by defining the solution of (\ref{AS-4}) as
$\mathbf{s}_{\varrho+1}$, the fifth level subproblem is solved to
find $\mathbf{p},\boldsymbol{\varphi},\boldsymbol{\delta}$ with a
given $\mathbf{s}_{\varrho}$:
\begin{align}\label{AS-6}
\max_{\mathbf{p}\in\mathbb{D}_{\mathbf{p}},\boldsymbol{\varphi}
\boldsymbol{\delta}}\Theta(\mathbf{p}_{\varrho},\tilde{\mathbf{p}}_{\varrho},\mathbf{s}_{\varrho},\boldsymbol{\theta}_{\varrho},\boldsymbol{\zeta}_{\varrho},\boldsymbol{\varphi}_{\varrho},{\mathbf{e}_h}_{\varrho}).
\end{align}

Let
$(\mathbf{p}_{\varrho},\tilde{\mathbf{p}}_{\varrho},\mathbf{s}_{\varrho},\boldsymbol{\theta}_{\varrho}
,\boldsymbol{\zeta}_{\varrho},{\mathbf{e}_h}_{\varrho})$ denote
the obtained solution at the $\varrho$-th iteration, which
should be used for the $\varrho+1$-th iteration. With a convergence
threshold $\epsilon_1$, the stop condition of alternate
optimization algorithm is then given by
\begin{align}
&|\Theta(\mathbf{p}_{\varrho},\tilde{\mathbf{p}}_{\varrho},\mathbf{s}_{\varrho},\boldsymbol{\theta}_{\varrho},\boldsymbol{\zeta}_{\varrho},\boldsymbol{\varphi}_{\varrho},{\mathbf{e}_h}_{\varrho})-\\\nonumber&
\Theta(\mathbf{p}_{\varrho+1},\tilde{\mathbf{p}}_{\varrho+1},\mathbf{s}_{\varrho+1},\boldsymbol{\theta}_{\varrho+1},\boldsymbol{\zeta}_{\varrho+1},\boldsymbol{\varphi}_{\varrho+1},{\mathbf{e}_h}_{\varrho+1})|\leq\epsilon_1.
\end{align}

We can also present a maximum allowed number $\Psi_1$ for
$\varrho_1$. Alternate optimization algorithm is illustrated in Table. \ref{Alternate Optimization Algorithm}. Furthermore, the following Theorem. \ref{Theorem 2}
can verify the convergence of the alternate optimization
algorithm.
\begin{theorem}\label{Theorem 2}
If (\ref{AS-1})
-(\ref{AS-6}) are solvable, in each
iteration, the sequence of each solution, i.e.,
$\{\Theta(\mathbf{p}_{\varrho},\mathbf{s}_{\varrho},\boldsymbol{\varphi}_{\varrho}
,\boldsymbol{\delta}_{\varrho})\}$, is monotonically decreasing.
\end{theorem}

\begin{proof}
Please see Appendix \ref{appendix-B}.
\end{proof}

\begin{table*}[h!]\caption{Alternate Optimization Algorithm}\label{Alternate Optimization Algorithm}
  \centering
  \begin{tabular}{lcc}
    \toprule
    \textbf{Algorithm 1}: Alternate Optimization Algorithm\\
    \midrule
    \textbf{Step1}: Select a starting point $(\mathbf{p}_0,\tilde{\mathbf{p}}_0,
    \mathbf{s}_0,\boldsymbol{\theta}_0,
    \boldsymbol{\zeta}_0,\boldsymbol{\varphi}_0
,\boldsymbol{\delta}_0)\in\mathbb{D}$, and Set iteration number
$\varrho=0$;\\
\textbf{Step2}: Compute $\Theta(\mathbf{p}_0,
\mathbf{s}_0,\boldsymbol{\varphi}_0
,\boldsymbol{\delta}_0)$;\\
\textbf{Repeat}\\
\textbf{Step3}: For fixed $\mathbf{p}_{\varrho}$, solve
(\ref{eq-worstcase-e}) to obtained the ${\mathbf{e}_h}_{\varrho+1}$ (Linear programming);\\
\textbf{Step4}: For the obtained ${\mathbf{e}_h}_{\varrho+1}$, and
fixed $\mathbf{p}_{\varrho},\tilde{\mathbf{p}}_{\varrho},
    \mathbf{s}_{\varrho},
    \boldsymbol{\zeta}_{\varrho},\boldsymbol{\varphi}_{\varrho}
,\boldsymbol{\delta}_{\varrho}$,
solve (\ref{AS-2}) to find $\boldsymbol{\theta}_{\varrho+1}$ (Linear programming);\\
\textbf{Step5}: For the obtained ${\mathbf{e}_h}_{\varrho+1}$,
$\boldsymbol{\theta}_{\varrho+1}$, and fixed $\mathbf{p}_{\varrho},
    \mathbf{s}_{\varrho},\boldsymbol{\varphi}_{\varrho} ,\boldsymbol{\delta}_{\varrho}$,
solve (\ref{AS-3}) to find $\tilde{\mathbf{p}}_{\varrho+1}$ and $\boldsymbol{\zeta}_{\varrho+1}$ (Convex programming);\\
\textbf{Step6}: For the obtained ${\mathbf{e}_h}_{\varrho+1}$,
$\boldsymbol{\theta}_{\varrho+1}$,
$\boldsymbol{\zeta}_{\varrho+1}$ and
$\tilde{\mathbf{p}}_{\varrho+1}$, and fixed
$\mathbf{p}_{\varrho},\boldsymbol{\varphi}_{\varrho}
,\boldsymbol{\delta}_{\varrho}$,
solve (\ref{AS-4}) to find $\mathbf{s}_{\varrho+1}$ (DC programming);\\
\textbf{Step7}: For the obtained ${\mathbf{e}_h}_{\varrho+1}$,
$\boldsymbol{\theta}_{\varrho+1}$,
$\boldsymbol{\zeta}_{\varrho+1}$,$\tilde{\mathbf{p}}_{\varrho+1}$
and $\mathbf{s}_{\varrho+1}$, solve (\ref{AS-6}) to find
$\mathbf{p}_{\varrho+1},\boldsymbol{\varphi}_{\varrho+1}$ and
$\boldsymbol{\delta}_{\varrho+1}$ (DC programming);\\
 \textbf{Step8}:Compute
$\Theta(\mathbf{p}_{\varrho+1},\mathbf{s}_{\varrho+1},\boldsymbol{\varphi}_{\varrho+1}
,\boldsymbol{\delta}_{\varrho+1})$;\\
\textbf{Step9}: $\varrho=\varrho+1$;\\
\textbf{Step10}: \textbf{If}
$|\Theta(\mathbf{p}_{\varrho},\mathbf{s}_{\varrho},
\boldsymbol{\varphi}_{\varrho} ,\boldsymbol{\delta}_{\varrho})-
\Theta(\mathbf{p}_{\varrho+1},\mathbf{s}_{\varrho+1},
\boldsymbol{\varphi}_{\varrho+1}
,\boldsymbol{\delta}_{\varrho+1})|\leq\epsilon_1$
or $\varrho>\Psi_1$ goto Step12, \textbf{else} goto Step3;\\
\textbf{End}\\
\textbf{Step11}: Return
$(\mathbf{p},\tilde{\mathbf{p}},\mathbf{s},\boldsymbol{\theta},
\boldsymbol{\zeta},\boldsymbol{\varphi}
,\boldsymbol{\delta})=(\mathbf{p}_{\varrho},\tilde{\mathbf{p}}_{\varrho},\mathbf{s}_{\varrho},
\boldsymbol{\theta}_{\varrho},
\boldsymbol{\zeta}_{\varrho},\boldsymbol{\varphi}_{\varrho}
,\boldsymbol{\delta}_{\varrho})$.\\
        \bottomrule
        \label{Alternate Optimization Algorithm}
  \end{tabular}
\vspace{-0.5cm}
\end{table*}
\subsubsection{Channel Uncertainty Problem} For minimizing the worst-case problem
 over $\boldsymbol{\mathcal{H}}$ in (\ref{AS-1}), we solve the following problem for each $t$, $b$, $u$, $m$, and $q$:
\begin{align}\label{eq-worstcase-e}
\max_{h^{nt}_{bq}\in\mathcal{H}^{nt}_{bq}}~&
R^{\text{E},mt}_{buq}\equiv\max_{h^{nt}_{bq}\in\mathcal{H}^{nt}_{bq}}~\frac{\sum_{n\in
		\mathcal{N}}\eta_{nm}s^{mt}_{bu}p^{mt}_{bu}
	h^{nt}_{bq}}{\hat{I}^{mt}_{buq}+(\sigma^{n}_q)^2}
\end{align}

We can rewrite (\ref{eq-worstcase-e}) as follows:
\begin{subequations}\label{eq-worstcase-e2}
\begin{align}
\max_{\textbf{h}^t}~&\frac{(\bar{\textbf{c}}^{mt}_{buq})^\text{T}{\textbf{h}}^t}{(\hat{\textbf{c}}^{mt}_{buq})^\text{T}{\textbf{h}}^t+(\sigma^{n}_q)^2},\\\text{s.t.}~&h^{nt}_{\acute{b}\acute{q}}\leq\tilde{h}^{nt}_{\acute{b}\acute{q}}+\varepsilon_{h^{nt}_{\acute{b}\acute{q}}},~\forall \acute{b},\acute{q},\acute{u},n\\
&\tilde{h}^{nt}_{\acute{b}\acute{q}}-\varepsilon_{h^{nt}_{\acute{b}\acute{q}}}\leq h^{nt}_{\acute{b}\acute{q}}, ~ \forall \acute{b},\acute{q},\acute{u},n,
\end{align}
\end{subequations}
where $\bar{\textbf{c}}^{mt}_{buq}$ is a vector of the same dimension as ${\textbf{h}}^t$ with all zero enry expect for $[\bar{\textbf{c}}^{mt}_{buq}]_{\acute{b}=b,\acute{q}=q,n}=\eta_{nm}s^{mt}_{bu}p^{mt}_{bu}, \forall n\in\mathcal{N}$, and  $\hat{\textbf{c}}^{mt}_{buq}$ is a vector of the same dimension as ${\textbf{h}}^t$ with $[\hat{\textbf{c}}^{mt}_{buq}]_{\acute{b},\acute{q},n}=\sum_{\acute{u}\in
	\mathcal{U}_{\acute{b}}}\eta_{nm}s^{mt}_{\acute{b}\acute{u}}p^{mt}_{\acute{b}\acute{u}}, \forall \acute{b}\neq b,\acute{q}\neq q, n$, and $[\hat{\textbf{c}}^{mt}_{buq}]_{\acute{b},\acute{q},n}=0$ for all other entries.
This problem has a linear fractional objective function, for which, Charnes-Cooper transformation can be used to reformulate it into the following linear programming optimization problem \cite{parsaeefard2017robust}:
\begin{subequations}\label{eq-worstcase-e3}
\begin{align}
\max_{\bar{\textbf{h}}^t,\mu}~&(\bar{\textbf{c}}^{mt}_{buq})^\text{T}\bar{\textbf{h}}^t,\\\text{s.t.}~~&(\hat{\textbf{c}}^{mt}_{buq})^\text{T}\bar{\textbf{h}}^t+\mu(\sigma^{n}_q)^2=1,\\&\bar{h}^{nt}_{\acute{b}\acute{q}}\leq\mu\tilde{h}^{nt}_{\acute{b}\acute{q}}+\mu\varepsilon_{h^{nt}_{\acute{b}\acute{q}}},~\forall \acute{b},\acute{q},\acute{u},n\\
&\mu\tilde{h}^{nt}_{\acute{b}\acute{q}}-\mu\varepsilon_{h^{nt}_{\acute{b}\acute{q}}}\leq \bar{h}^{nt}_{\acute{b}\acute{q}}, ~ \forall \acute{b},\acute{q},\acute{u},n,
\end{align}
\end{subequations}
where $\bar{\textbf{h}}^t={\textbf{h}}^t/\mu$, $\bar{\textbf{h}}^t\succeq \textbf{0}$ and $\mu> 0$.
Problem (\ref{eq-worstcase-e3}) can now be efficiently solved using interior-point based methods
by some off-the-shelf convex optimization toolboxes, e.g., CVX.

\subsubsection{Content Placement} 
A linear programming (LP) with respect to $\boldsymbol{\theta}$ for the content placement problem can be obtained. This problem can be easily solved by existing LP available standard optimization softwares such as CVX with the internal
solver MOSEK \cite{mokari2016limited,michael2012matlab}. 

\subsubsection{Backhaul Power and Subcarrier Allocation} 
The optimization problem is still a mixed-integer non-convex
programming with respect to $\boldsymbol{\zeta}$ and $\tilde{\boldsymbol{p}}$, which is difficult to tackle. To make this problem tractable,
we first relax each $\boldsymbol{\zeta}$ to a continuous interval, i.e., $\boldsymbol{\zeta}\in$[0,1]. Further, new variables $\textbf{x}=\boldsymbol{\zeta}\tilde{\boldsymbol{p}}$ is defined to replace $\tilde{\boldsymbol{p}}$. Then, we can transform the nonconvex optimization problem into the convex one. This problem can be easily solved by available standard optimization softwares such as CVX with the internal
solver MOSEK \cite{mokari2016limited}. Note that this relaxation is called time sharing which shows the time percentage that each subcarrier should be used \cite{tao2008resource}.

\subsubsection{Access Power and Codebook Allocation} The optimization problem is still non-convex with respect to $\textbf{p}$ and $\textbf{s}$. The difficulty of solution comes from the non-convexity of both objective function and secrecy rate constraint. There is no standard approach to solve such
a non-convex problem. Therefore, we exploit DC and fractional programming in the next sections to transform it into a tractable problem. In the following, we develop a solution for power allocation optimization problem and we remark that this solution can be developed for code assignment in the same way.

\subsection{Difference-of-Two-Concave-Functions (D.C.) Approximation}
Due to the non-convexity of (\ref{eq-EEE}), the optimization
problem (\ref{eq-worstcase-p-1234}) is still difficult to solve.
The standard D.C. optimization problem can be written as
$\min_{\mathbf{x}}\{F(\mathbf{x})=F_1(\mathbf{x})-F_2(\mathbf{x})\}$
where $F_1$ and $F_2$ are two convex components with convex feasible domain. This problem can be solved
iteratively by solving a sequential convex program as follows:
\begin{equation}
\min_{\mathbf{x}}\{F_1(\mathbf{x})-F_2(\mathbf{x}_{\varrho})-\langle
\nabla
F_2(\mathbf{x}_{\varrho}),\mathbf{x}-\mathbf{x}_{\varrho}\rangle\},
\end{equation}
at each iteration, where $\mathbf{x}_{\varrho}$ is the optimal
solution of the $\varrho^{\text{th}}$ iteration used for the
$(\varrho+1)^{\text{th}}$ iteration and $\nabla F_2(\mathbf{x})$ is
the gradient of $F_2(\mathbf{x})$ evaluated at
$\mathbf{x}_{\varrho}$. By the sequential convex approximation, DC subproblems are equivalently reformulated as:
\begin{subequations}\label{eq-worstcase-p-1234-1}
\begin{align}
\max_{\mathbf{p},\boldsymbol{\varphi} ,\boldsymbol{\delta}}
&~\Theta(\mathbf{p}_{\varrho},\tilde{\mathbf{p}}_{\varrho},\mathbf{s}_{\varrho},\boldsymbol{\theta}_{\varrho},\boldsymbol{\zeta}_{\varrho},\boldsymbol{\varphi}_{\varrho},{\mathbf{e}_h}_{\varrho}),\\\nonumber
\text{s.t.}~&-\left(R^{\text{E},mt}_{buq2}-R^{\text{E},mt}_{buq1}-\left\langle\nabla
R^{\text{E},t}_{bu1},
p^{mt}_{bu}-p^{mt}_{bu}(\varrho)\right\rangle
\right)\leq\\&\varphi^{mt}_{bu},\forall m\in\mathcal{M},
t\in\mathcal{F},b\in\mathcal{B},u\in\mathcal{U}_b,q\in\mathcal{Q},
\\\nonumber&(\ref{eq-P-constraint-1}),(\ref{eq-energy-2}),(\ref{eq-energy-3}),(\ref{eq-EEE-1}),(\ref{eq-EEE-2}),(\ref{eq-EEE-3}),
\end{align}
\end{subequations}
where
\begin{align}
& R^{\text{E},mt}_{buq1}=\log_2\left(\sum_{b\in
\mathcal{B}}\sum_{u\in
\mathcal{U}_b}\sum_{n\in
\mathcal{N}}\left(\eta_{nm}s^{mt}_{bu}
p^{mt}_{bu}h^{nt}_{bq}+(\sigma^{n}_q)^2\right)\right)\\
&R^{\text{E},mt}_{bu2}=\log_2\left(\sum_{\acute{b}\in
\mathcal{B}\setminus\{b\}}\sum_{\acute{u}\in
\mathcal{U}_{\acute{b}}}\sum_{n\in
\mathcal{N}}\left(\eta_{nm}s^{mt}_{\acute{b}\acute{u}}
p^{mt}_{\acute{b}\acute{u}}
h^{nt}_{\acute{b}q}+(\sigma^{n}_q)^2\right)
\right).
\end{align}

We first express $R^{\text{E},mt}_{buq}$ in a D.C. form as:
\begin{equation}\label{E}
R^{\text{E},mt}_{buq}=-(R^{\text{E},mt}_{buq2}-R^{\text{E},mt}_{buq1}).
\end{equation}
Based on (\ref{E}), the gradient $\nabla R^{\text{E},mt}_{bu1}$
with respect to $\mathbf{p}$ is given by
\begin{align}\label{eq-53}
&\nabla R^{\text{E},mt}_{bu1}=\frac{\partial
R^{\text{E},mt}_{bu1}}{\partial p^{mt}_{bu}}=
\frac{\partial R^{\text{E},mt}_{buq1}}{\partial
p^{mt}_{bu}}=\\\nonumber&\frac{s^{mt}_{bu}}{\ln2} \frac{\sum_{n\in
\mathcal{N}}\left(\eta_{nm}h^{nt}_{bq}\right)} {\sum_{b\in
\mathcal{B}}\sum_{u\in \mathcal{U}_b}\sum_{n\in
\mathcal{N}}\left(\eta_{nm}s^{mt}_{bu}p^{mt}_{bu}h^{nt}_{bq}+(\sigma^{n}_q)^2\right)}.
\end{align}

\begin{theorem}
The sequence $R^{\text{E},mt}_{buq}$ derived from the D.C.
programming algorithm is monotonically decreasing.
\end{theorem}

\begin{proof}
Please see Appendix \ref{appendix-C}.
\end{proof}

\subsection{Fractional Programming}
The objective function in (\ref{eq-worstcase-p-1234}) is
non-convex. The form of (\ref{eq-worstcase-p-1234}) can be classified into the nonlinear fractional programming \cite{Dinkelbach67}. Therefore, after replacing nonconvex constraints by convex constraints using the D.C. method in the previous section, the Dinkelbach's algorithm use to solve convex fractional programming. We define the maximum objective functions
$(\chi^{\mathbf{p}})^*$ of the considered system as:
\begin{align}\label{eq-zeta}\nonumber
&(\chi^{\mathbf{p}})^*=\max_{\mathbf{p},\boldsymbol{\varphi},\boldsymbol{\delta}}\frac{\sum_{m\in
\mathcal{M}}\sum_{t\in F}\sum_{b\in \mathcal{B}}\sum_{u\in
\mathcal{U}_b}\left\{\delta^{mt}_{bu}
+R^{\text{D},mt}_{bu1}\right\}} {\sum_{t\in \mathcal{F}}\sum_{b\in
\mathcal{B}}\sum_{m\in \mathcal{M}}\sum_{u\in
\mathcal{U}_b}s^{mt}_{bu}p^{mt}_{bu}}\\&=\frac{\Xi_{\text{Num}}(\mathbf{p},\mathbf{s},\boldsymbol{\varphi},\boldsymbol{\delta})}
{\Xi_{\text{Den}}(\mathbf{p},\mathbf{s},\boldsymbol{\varphi},\boldsymbol{\delta})}.
\end{align}

We are now ready to introduce the following theorem for $\chi^{\mathbf{p}}$. 
\begin{theorem}
The maximum value of $(\chi^{\mathbf{p}})^*$ is achieved if and only if
\begin{align}\label{eq-zeta-2}
\max_{\mathbf{p},\boldsymbol{\varphi}
,\boldsymbol{\delta}}~&
\left\{{\Xi_{\text{Num}}(\mathbf{p},\mathbf{s},\boldsymbol{\varphi},\boldsymbol{\delta})}-(\chi^{\mathbf{p}})^*
{\Xi_{\text{Dem}}(\mathbf{p},\mathbf{s},\boldsymbol{\varphi},\boldsymbol{\delta})}\right\}\\\nonumber&
={\Xi_{\text{Num}}(\mathbf{p}^*,\mathbf{s},\boldsymbol{\varphi}^*,\boldsymbol{\delta}^*)}-(\chi^{\mathbf{p}})^*
{\Xi_{\text{Num}}(\mathbf{p}^*,\mathbf{s},\boldsymbol{\varphi}^*,\boldsymbol{\delta}^*)}=0.
\end{align}
For $\Xi_{\text{Num}}(\mathbf{p},\mathbf{s},\boldsymbol{\varphi},\boldsymbol{\delta})\geq0$ and
$\Xi_{\text{Dem}}(\mathbf{p},\mathbf{s},\boldsymbol{\varphi},\boldsymbol{\delta})>0$, where
\begin{equation}\label{eq-20}
\max_{\mathbf{p},\boldsymbol{\varphi}
,\boldsymbol{\delta}}~\left\{{\Xi_{\text{Num}}(\mathbf{p},\mathbf{s},\boldsymbol{\varphi},\boldsymbol{\delta})}-\chi^{\mathbf{p}}
{\Xi_{\text{Dem}}(\mathbf{p},\mathbf{s},\boldsymbol{\varphi},\boldsymbol{\delta})}\right\},
\end{equation}
is defined as a parametric program with parameter $\chi^{\mathbf{p}}$. 
\end{theorem}

\begin{proof}
Please refer to \cite{Dinkelbach67,Schaible76}.
\end{proof}

By the Dinkelbach's method [23] with a initial value $\chi^{\mathbf{p}}_0$ of
$\chi^{\mathbf{p}}$, (\ref{eq-20}) can be solved iteratively by solving the following problem:
\begin{equation}\label{eq-21}
\max_{\mathbf{p},\boldsymbol{\varphi}
,\boldsymbol{\delta}}~{\Xi_{\text{Num}}(\mathbf{p},\mathbf{s},\boldsymbol{\varphi},\boldsymbol{\delta})}-\chi^{\mathbf{p}}_{\varrho}
{\Xi_{\text{Dem}}(\mathbf{p},\mathbf{s},\boldsymbol{\varphi},\boldsymbol{\delta})},
\end{equation}
with a given $\chi^{\mathbf{p}}_{\varrho}$ at the $\varrho^{\text{th}}$
iteration, where $\varrho$ is the iteration index.
$\chi^{\mathbf{p}}_{\varrho}$ can be explained as the secure EE obtained at the
previous iteration. In (\ref{eq-21}), the maximization problem is
equivalent to
\begin{equation}\label{eq-22}
\min_{\mathbf{p},\boldsymbol{\varphi}
,\boldsymbol{\delta}}~{\chi^{\mathbf{p}}_{\varrho}
{\Xi_{\text{Dem}}(\mathbf{p},\mathbf{s},\boldsymbol{\varphi},\boldsymbol{\delta})-\Xi_{\text{Num}}(\mathbf{p},\mathbf{s},\boldsymbol{\varphi},\boldsymbol{\delta})}}.
\end{equation}

Let
$\mathbf{p}(\chi^{\mathbf{p}}_{\varrho}), \boldsymbol{\varphi}(\chi^{\mathbf{p}}_{\varrho})$ and $\boldsymbol{\delta}(\chi^{\mathbf{p}}_{\varrho})$ denote the solution of (\ref{eq-23})
for a given $\chi^{\mathbf{p}}_{\varrho}$. After each iteration,
$\chi^{\mathbf{p}}_{\varrho}$ should be updated by
\begin{equation}\label{eq-23}
\chi^{\mathbf{p}}_{\varrho+1}=\frac
{\Xi_{\text{Dem}}(\mathbf{p}(\chi^{\mathbf{p}}_{\varrho}),\boldsymbol{\varphi}(\chi^{\mathbf{p}}_{\varrho}),\boldsymbol{\delta}(\chi^{\mathbf{p}}_{\varrho}),\mathbf{s})}
{\Xi_{\text{Num}}(\mathbf{p}(\chi^{\mathbf{p}}_{\varrho}),\boldsymbol{\varphi}(\chi^{\mathbf{p}}_{\varrho}),\boldsymbol{\delta}(\chi^{\mathbf{p}}_{\varrho}),\mathbf{s})}.
\end{equation}

The iteration process will be stopped when (\ref{eq-zeta-2}) is
satisfied. In practice, we define the terminated condition of the
iterative process as:
\begin{align}
&\Big|\chi_{\varrho}\Xi_{\text{Dem}}(\mathbf{p}(\chi^{\mathbf{p}}_{\varrho}),\boldsymbol{\varphi}(\chi^{\mathbf{p}}_{\varrho}),\boldsymbol{\delta}(\chi^{\mathbf{p}}_{\varrho}),\mathbf{s})-\\\nonumber&
\Xi_{\text{Num}}(\mathbf{p}(\chi^{\mathbf{p}}_{\varrho}),\boldsymbol{\varphi}(\chi^{\mathbf{p}}_{\varrho}),\boldsymbol{\delta}(\chi^{\mathbf{p}}_{\varrho})),\mathbf{s})\Big|\leq\epsilon_3,
\end{align}
with a small convergence tolerance $\epsilon_3>0$. The algorithm
of fractional programming is clarified in Algorithm 2, where
$\Psi_3$ is the maximum allowed number of iterations considering
the computational time. We use the fractional programming Dinkelbach's algorithm for the convexified problem (\ref{eq-worstcase-p-1234-1}).

\begin{theorem}
If problems (\ref{eq-worstcase-p-1234-11}) are solvable, the sequence
$\{\chi_{\varrho}\}$ obtained by Algorithm 1 has the following
properties: 1) $\chi^{\mathbf{p}}_{\varrho+1}>\chi^{\mathbf{p}}_{\varrho}$; 2)
$\lim_{\varrho\rightarrow\infty}\chi^{\mathbf{p}}_{\varrho}=(\chi^{\mathbf{p}})^*$.
\end{theorem}

\begin{proof}
Please refer to \cite{Dinkelbach67}.
\end{proof}

Based on the fractional programming, subproblems
(\ref{eq-worstcase-p-1234-11}) are associated with a parametric program problem stated as follows:
\begin{subequations}\label{eq-worstcase-p-1234-11}
\begin{align}
\max_{\mathbf{p},\boldsymbol{\varphi} ,\boldsymbol{\delta}}
&~\chi^{\mathbf{p}}\sum_{t\in \mathcal{F}}\sum_{b\in \mathcal{B}}\sum_{m\in
\mathcal{M}}\sum_{u\in
\mathcal{U}_b}s^{mt}_{bu}p^{mt}_{bu}-\\\nonumber&\sum_{m\in
\mathcal{M}}\sum_{t\in F}\sum_{b\in \mathcal{B}}\sum_{u\in
\mathcal{U}_b}\left\{\delta^{mt}_{bu}
+R^{\text{D},mt}_{bu1}\right\},\\\nonumber
\text{s.t.}~&-\left(R^{\text{E},mt}_{buq2}-R^{\text{E},mt}_{buq1}-\left\langle\nabla
R^{\text{E},t}_{bu1},
p^{mt}_{bu}-p^{mt}_{bu}(\varrho)\right\rangle
\right)\leq\varphi^{mt}_{bu},\\&\forall m\in\mathcal{M},
t\in\mathcal{F},b\in\mathcal{B},u\in\mathcal{U}_b,q\in\mathcal{Q},
\\\nonumber&(\ref{eq-P-constraint-1}),(\ref{eq-energy-2}),(\ref{eq-energy-3}),
(\ref{eq-EEE-1}),(\ref{eq-EEE-2}),(\ref{eq-EEE-3}).
\end{align}
\end{subequations}

We propose an iterative algorithm (known as the Dinkelbach method
\cite{Dinkelbach67}) for solving (\ref{eq-worstcase-p-1234-11}) with an equivalent objective
function. The proposed algorithm to obtain power allocation policy $\textbf{p}$ is summarized in Table. \ref{Fractional Programming Algorithm}. The
convergence to the appropriate energy efficiency is guaranteed. Note that similar algorithm can be used to obtain code allocation policy $\textbf{s}$.

\begin{table*}[h!]\caption{Fractional Programming Algorithm}\label{Fractional Programming Algorithm}
  \centering
  \begin{tabular}{lcc}
    \toprule
    \textbf{Algorithm 3}: Fractional Programming Algorithm\\
    \midrule
    \textbf{Step1}: Initialize the maximum number of iterations $\Psi_3$ and the
maximum tolerance $\epsilon_3$;\\
\textbf{Step2}: Choose an initial value $\chi^{\mathbf{p}}_0$ and set iteration index $\varrho=0$;\\
\textbf{Repeat}\\
\textbf{Step3}: Solve  problem 
(\ref{eq-worstcase-p-1234-11}) for a given $\chi^{\mathbf{p}}_{\varrho}$ and
obtain power allocation policy $\mathbf{p}(\chi^{\mathbf{p}}_{\varrho})$
(Convex programming);\\
\textbf{Step5}: Update $\chi^{\mathbf{p}}_{\varrho}$ by (\ref{eq-23}) to obtain
$\chi^{\mathbf{p}}_{\varrho+1}$;\\
\textbf{Step6}:$\varrho=\varrho+1$;\\
\textbf{Step7}: If
$|\chi^{\mathbf{p}}_{\varrho}\Xi_{\text{Den}}(\mathbf{p}(\chi^{\mathbf{p}}_{\varrho}),\mathbf{s})-
\Xi_{\text{Num}}(\mathbf{p}(\chi^{\mathbf{p}}_{\varrho}),\mathbf{s})|<\epsilon_3$
or $\varrho>\Psi_3$ goto Step7,
\textbf{else} goto Step3;\\
\textbf{Step8}: Return $\mathbf{p}^*=\mathbf{p}(\chi^{\mathbf{p}}_{\varrho-1})$, $(\chi^{\mathbf{p}})^*=\chi^{\mathbf{p}}_{\varrho}$.\\
   \textbf{End}\\
    \bottomrule
  \end{tabular}
\end{table*}

To solve the primary optimization problem, the main optimization
problem is decomposed into several subproblems, with each
subproblem being in a hierarchical order of the main problem. Depending on different methods to solve each subproblem, the computational complexity of the proposed algorithm is analyzed in Section \ref{computationalanalysis}.
	
\section{Analysis of Computational Complexity of Proposed Algorithm}\label{computationalanalysis}
	To solve the primary optimization problem, the main optimization
	problem is decomposed into several subproblems, with each
	subproblem being in a hierarchical order of the main problem. The fast gradient algorithm
	can be used to solve the inner convex subproblems
	\cite{richter2012computational}. Then, the convergence of fast
	gradient algorithm can be written as follows
	\cite{richter2012computational}:
	\begin{equation}
	\varrho^{\iota}_{\tau}=\mathcal{O}(1)
	\min\left\{\sqrt{\frac{l_{\iota}}{\xi_{\iota}}}\ln\left(\frac{l_{\iota}}{\tau}\right),
	\sqrt{\frac{l_{\iota}}{\xi_{\iota}}}\right\},\iota=1,2,
	\end{equation}
	where $\tau$ is convergence tolerance. $l_1$ and $l_2$ are defined
	as the Lipschitz constants by which the Lipschitz conditions are
	satisfied with the gradients $\nabla
	R^{\text{E},t}_{buq,1}(\mathbf{p},\mathbf{s})$ in subproblem
	(\ref{AS-2}) and (\ref{AS-3}), respectively. Moreover, $\xi_1$ and
	$\xi_2$ denote the convexity parameters to satisfy strong
	convexity of $f_1$ and $f_2$, respectively. In other words, the
	fast gradient method is used to solve optimization problem
	(\ref{AS-2}) and (\ref{AS-3}) in $\varrho^{1}_{\tau}$ and
	$\varrho^{2}_{\tau}$ iterations, respectively. The iteration
	numbers $\varrho_1, \varrho_2$, and $\varrho_3$ are corresponding
	to the convergence tolerance $\epsilon_1, \epsilon_2$, and
	$\epsilon_3$ related to subalgorithms. Then, the overall
	computational complexity associated with the proposed algorithm is
	dominated by
	$\varrho_1\varrho_3(\varrho^1_2\varrho^{1}_{\tau}+\varrho^2_2\varrho^{2}_{\tau})$,
	where $\varrho^1_2$ and $\varrho^2_2$ are the iteration numbers of
	DC programming algorithm. The computational complexity of the proposed solution of the content placement problem is equal to $\mathcal{O}(I_{\theta}B^{3.5}K^{3.5})$ \cite{ben2001lectures} where $I_{\theta}$ is the
	expected iteration numbers. The computational complexity of the our proposed MFCD problem is equal to $\mathcal{O}(I_{\beta}F^{3.5}K^{3.5})$ \cite{ben2001lectures} where $I_{\beta}$ is the
	expected iteration numbers. The computational complexity of the proposed backhaul power and subcarrier allocation problems are equal to $\mathcal{O}(\hat{I}_{\tilde{p}}NFB
	\log_2(1/\epsilon_{\tilde{p}}))$ and $\mathcal{O}(\hat{I}_{\zeta}NFB
	\log_2(1/\epsilon_{\zeta}))$ where $\epsilon_{\tilde{p}}$, $\hat{I}_{\tilde{p}}$, 
	$\epsilon_{\zeta}$ and $\hat{I}_{\zeta}$ are the maximum error tolerance and expected iteration numbers, respectively.

By the proposed algorithm, the main problem, which is NP-hard, can be solved in polynomial time. Therefore, the complexity of the proposed algorithm compared to the original problem is far less and manageable. To manage the complexity, Graphics Processing Unit (GPU) can be utilized. By exploiting GPU instead of central processing unit (CPU), new processing methods to accelerate the processing time can be used. In \cite{naderazmi2018arxiv}, a new framework to accelerate the iterative-based resource allocation by using ASM and SCA has been devised. By using GPU-based resource allocation, the processing time speed-up of about 1500 times compared to CPU-based methods can be achieved. Due to the space limitation, we omit the study of GPU-based version of the proposed algorithms and leave it as future work.

\section{Simulation Results}\label{simulationsresults}
For simulations, we consider a multi-cell downlink SCMA system
where $U$ users are randomly distributed in an area of circle
with the radius of 1 km for each BS as the center. The number of
users in circle area of $b^{\text{th}}$ BS is set to $U_b=4,\forall b$, and
the total number of subcarriers and codebooks are set to 8 and 28,
respectively. The bandwidth of each subcarrier
is 180 kHz \cite{etsi2011136}. The channels between the MBS and its users and SBS and its users are generated with a normalized Rayleigh fading
component and a distance-dependent path loss in urban and suburban areas, modeled as
$PL(dB)=128.1+37.6\log_{10}(d)+X$ and $PL(dB)=38+30\log_{10}(d)+X$, respectively \cite{etsi2011136}, where $d$ is the distance from user to BS in kilometers and $X $ is 8 dB log-normal shadowing. We set the frame duration to $T=0.01$ s \cite{castiglione2014energy,zhang2015delay}. The noise power, $(\sigma^n_u)^2=(\sigma^n_b)^2=\sigma^2,\forall u,b,n$ is set to $-125$ dBm. We set $D=2$ and $\eta_{nm}=0.5,\forall n,m$ for SCMA
\cite{nikopour2013sparse}. We set the amount of harvested energy
per arrival to $\rho^t_b=\rho=0.8, \forall t, b$ J,  $\Gamma_b=\Gamma=0.1,\forall b$ and users request contents by normal random generator. In the most popular caching case, the most popular contents is cached at each BS until its storage is full. In this case, the content popularity is modeled as the Zipf distribution with Zipf parameter equals to 0.8. Simulation results are obtained by averaging over 1000 simulation runs.

\subsection{Effect of Maximum Allowable Backhaul Transmission Power}
In this part, we obtain the backhaul rate for different values of
backhaul transmission power with different values of $\alpha$. The simulation results are compared for different caching scenarios such as no caching, random caching, most popular caching and the proposed caching methods. In no caching case, no contents are stored by any BS. Hence, all the
requested contents are served by the core network over the backhaul links \cite{park2016joint,stephen2016green,
	hsu2016resource}. In the random caching strategy, the contents  are randomly cached by BSs until storage of BSs is full. Content popularity does not matter in this strategy. In the most popular caching strategy, each BS caches the most popular
contents until its storage is full \cite{stephen2016green,park2016joint}. The results are reported in Fig.
\ref{Fig_R_backhual_vs_p_diff_alpha}. As can be seen, for a fixed transmit power, when $\alpha$ is increased, the resulting backhaul rate increases. As can be seen from this figure, utilizing the caching strategies can reduce backhaul traffic compared
to the no caching scheme. Our proposed caching strategy has nearly 43\%, 23.4\% and 18.5\% performance gain in terms
	of backhaul rate reduction compared to the no caching scheme for different values of $\alpha=$1, 2, and 3, respectively. It is also notable that the most popular caching strategy causes
more reduction in the backhaul traffic, compared to the random caching scheme. However, when all the caching placement are done jointly with the allocation of other network resources, the network performance improves dramatically. This improvement is due to the fact that content placement is done according to network conditions and resources. Besides, as shown in Fig. \ref{Fig_R_backhual_vs_p_diff_alpha}, our caching scheme reduces the total backhaul rate close to almost 11\% compared to the most popular caching strategy.
\begin{figure}[h]
	\begin{center}
		\subfigure[]{\label{Fig_R_backhual_vs_p_diff_alpha}
			\includegraphics[width=3.3 in]
			{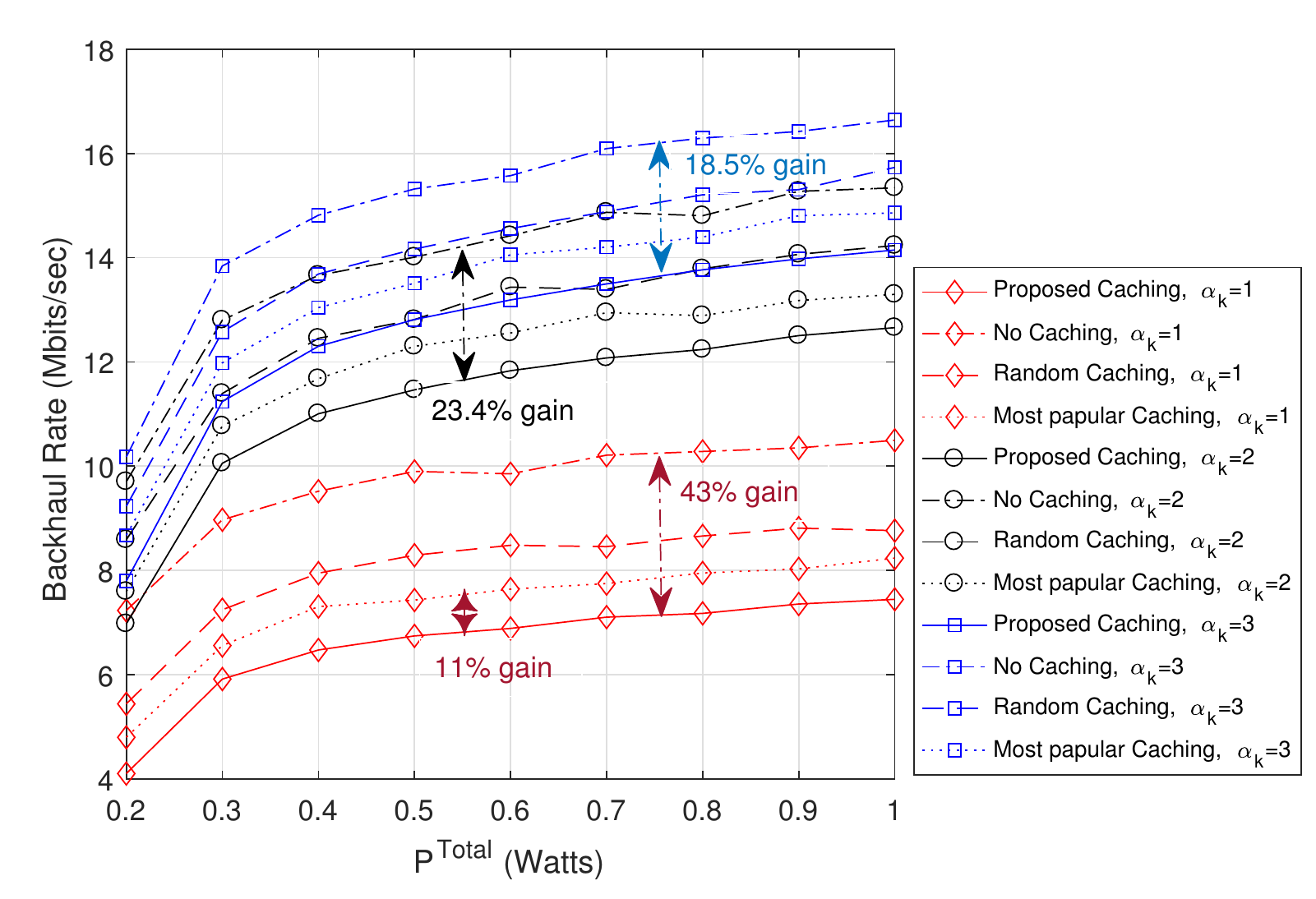}}
		\subfigure[]{\label{Fig_Energy_efficiency_vs_E_b_diff_U}
			\includegraphics[width=3 in]
			{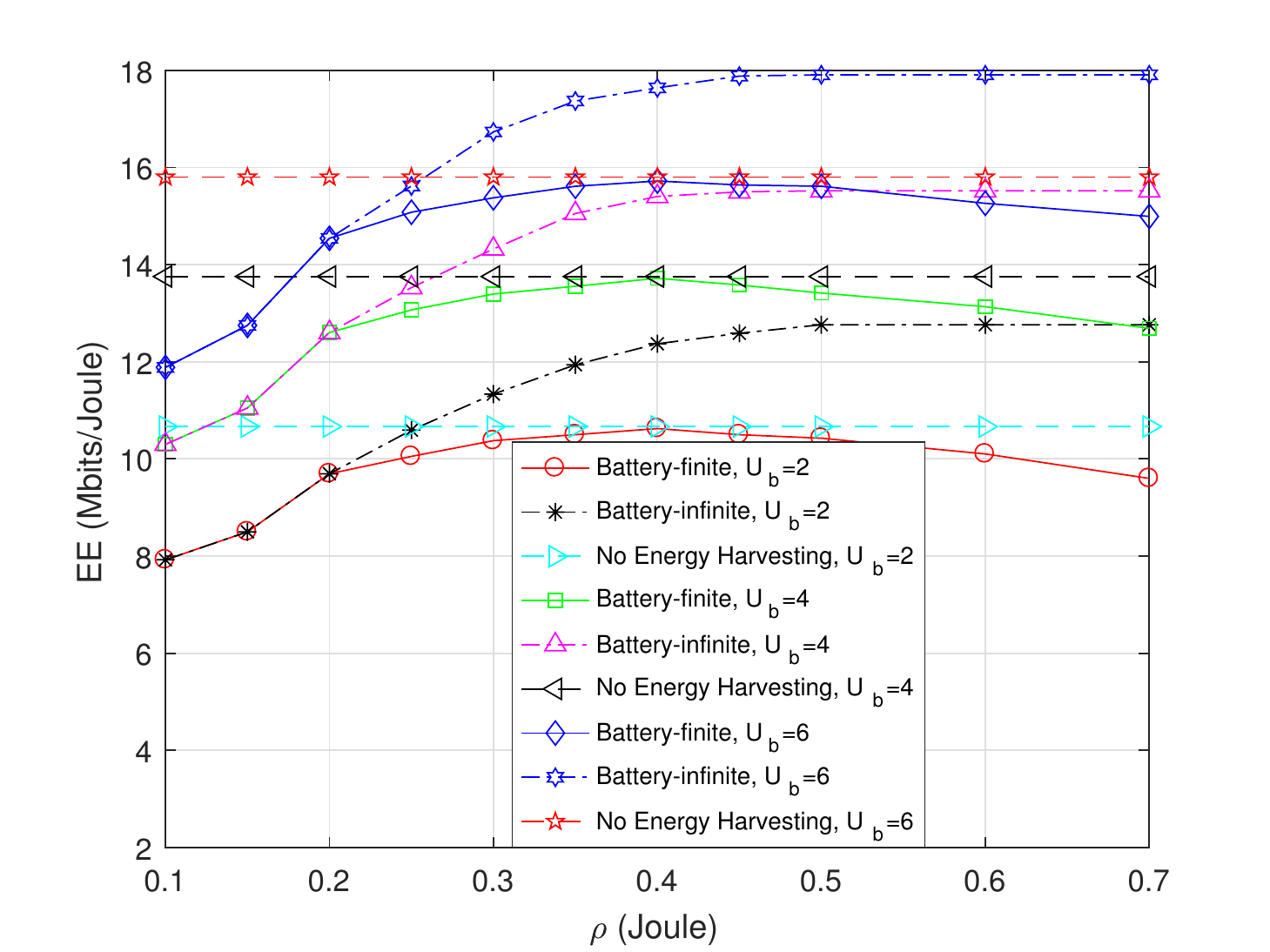}}
		\caption{(a) Backhaul rate, $\tilde{R}$, \emph{vs}.
			the maximum backhaul transmit power constraint, $\tilde{P}^{\text{Total}}$ for the SFCD scenario.
			System parameters are: $B=2, M=28, F=2, U=4, N=8, Q=1, K=6, D=2, T=0.01$ s, $\forall k$, $V_b=10$ Mbits,  $\forall b$ , $E^{\text{max}}_b=2$ Joule, $\forall b$, $\rho=0.5$ Joule, $\varepsilon_{h}=0.5$. (b) Energy efficiency, $EE$, \emph{vs}.
			the harvested energy per arrival, $\rho$ for the SFCD scenario. System parameters are: $B=2, M=28, F=2, U=2, N=8, Q=2, K=6, D=2, T=0.01$ s,
			$\alpha_k=1$ Mbits, $\forall k$, $V_b=10$ Mbits, $\forall b$, $\tilde{P}^{\text{Total}}=0.1$ Watts, $E^{\text{max}}_b=2$ Joule, $\forall b$, $\varepsilon_{h}=0.5$.}
	\end{center}
\end{figure}

\subsection{Effect of Energy Harvesting}
Fig. \ref{Fig_Energy_efficiency_vs_E_b_diff_U} shows the EE as a function of harvested energy per arrival for the SFCD scenario. We compare different EH strategy in terms of EE. In general, by increasing the EH value, the EE is also increased. For larger number of users, the EE is increased. In other words, for the small number of users, there is sufficient power resources, therefore by increasing users, the EE is also increased as shown in Fig. \ref{Fig_Energy_efficiency_vs_E_b_diff_U}. However, for too more users, the power
resource will be exhausted and thus some users can not access to network. Even so, due to multiuser diversity, the EE will still increase. For limited battery, due to overflow conditions (\ref{eq-energy-3}), the stored energy must be used such that there is enough capacity in the battery for newly arrived energy. In this regard, increasing the value of $\rho$ will increase the energy efficiency at first, but with further increasing $\rho$, the energy efficiency decreases. This is because, from the energy efficiency point of view, the energy consumption would be limited to the amount which maximizes the bit-per-joule quantity.   However, for unlimited battery, with increasing $\rho$, the energy efficiency increases at first, and by further increasing $\rho$, the energy efficiency becomes constant since no more energy would be consumed as all the arriving energy could be stored in the battery.

\subsection{Effect of File Spliting}
Fig. \ref{fig_EE_vs_Packet_Spliting} shows EE as a function of the harvested energy per arrival, $\rho$ for the SFCD and MFCD scenarios As can be seen from Fig. \ref{fig_EE_vs_Packet_Spliting}, the MFCD scheme outperforms the  SFCD scenario. In the EE communication networks, due to random energy arrivals, there may be not enough energy to transmit file that has big size in the SFCD scheme. In contrast, in the MFCD schemes, file is splitted into the several small size files which can be transmitted in the suitable frames to increase EE. In the uniform file splitting, the file is uniformly splitted into several smaller files with the same size. This scheme has better performance than the SFCD scheme, However, we can improve the network performance by using the our proposed method. In the our proposed MFCD scheme, the best size of each splitted file is obtained to enhance the network performance. This figure  also shows that by reducing the size of file, the distance between the graphs for the three scenarios decreases. As seen, the MFCD-proposed file splitting and MFCD-uniform file splitting have closed to almost 9.4\% and 6\% performance gain in terms of EE compared to SFCD scheme, respectively.

\begin{figure}[h]
	\begin{center}
		\subfigure[]{\label{fig_EE_vs_Packet_Spliting}\includegraphics[width=3.in]
			{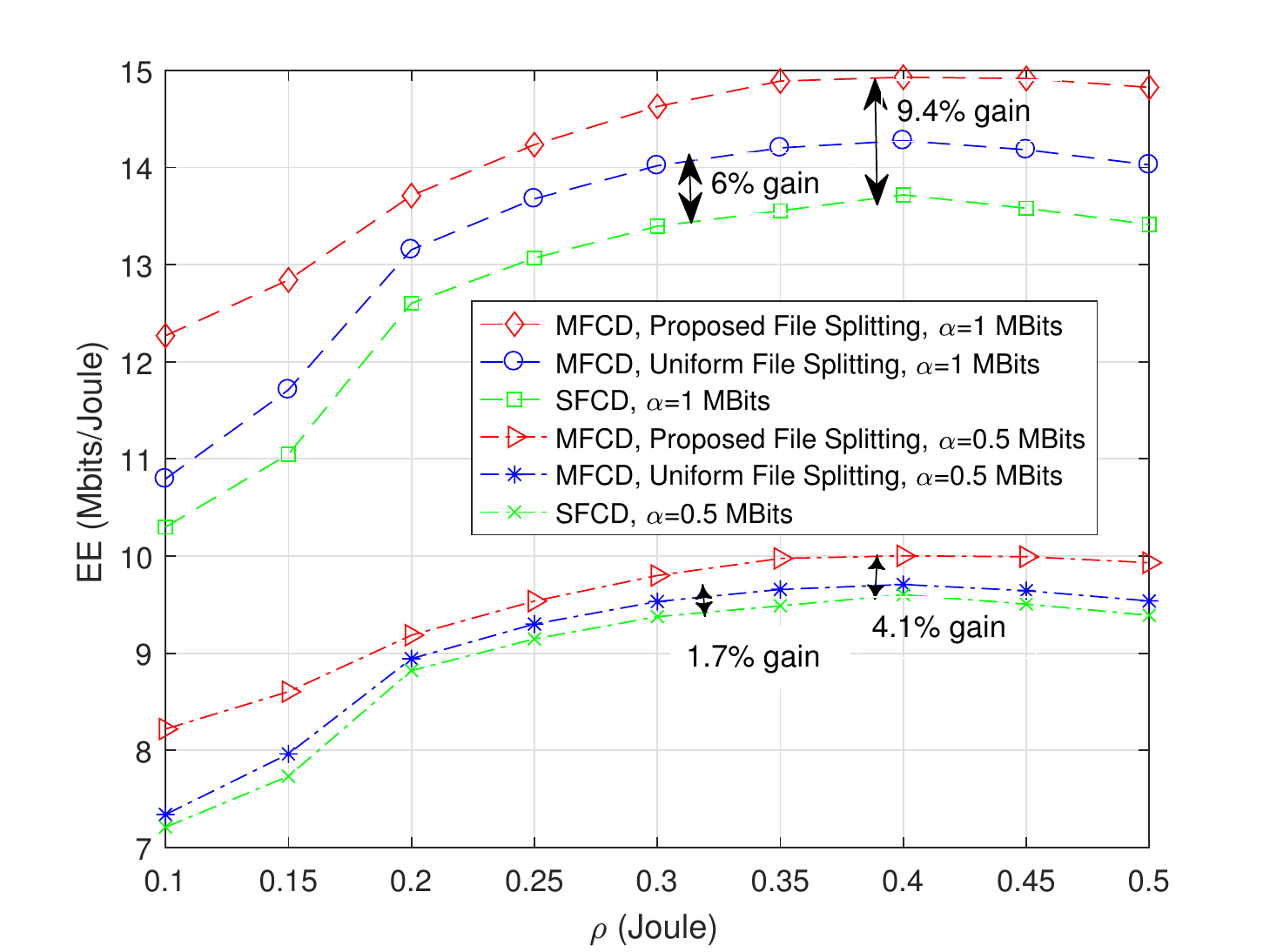}}
		\subfigure[]{\label{Fig_Outage_probability_vs_F_diff_alpha}\includegraphics[width=3.in]
			{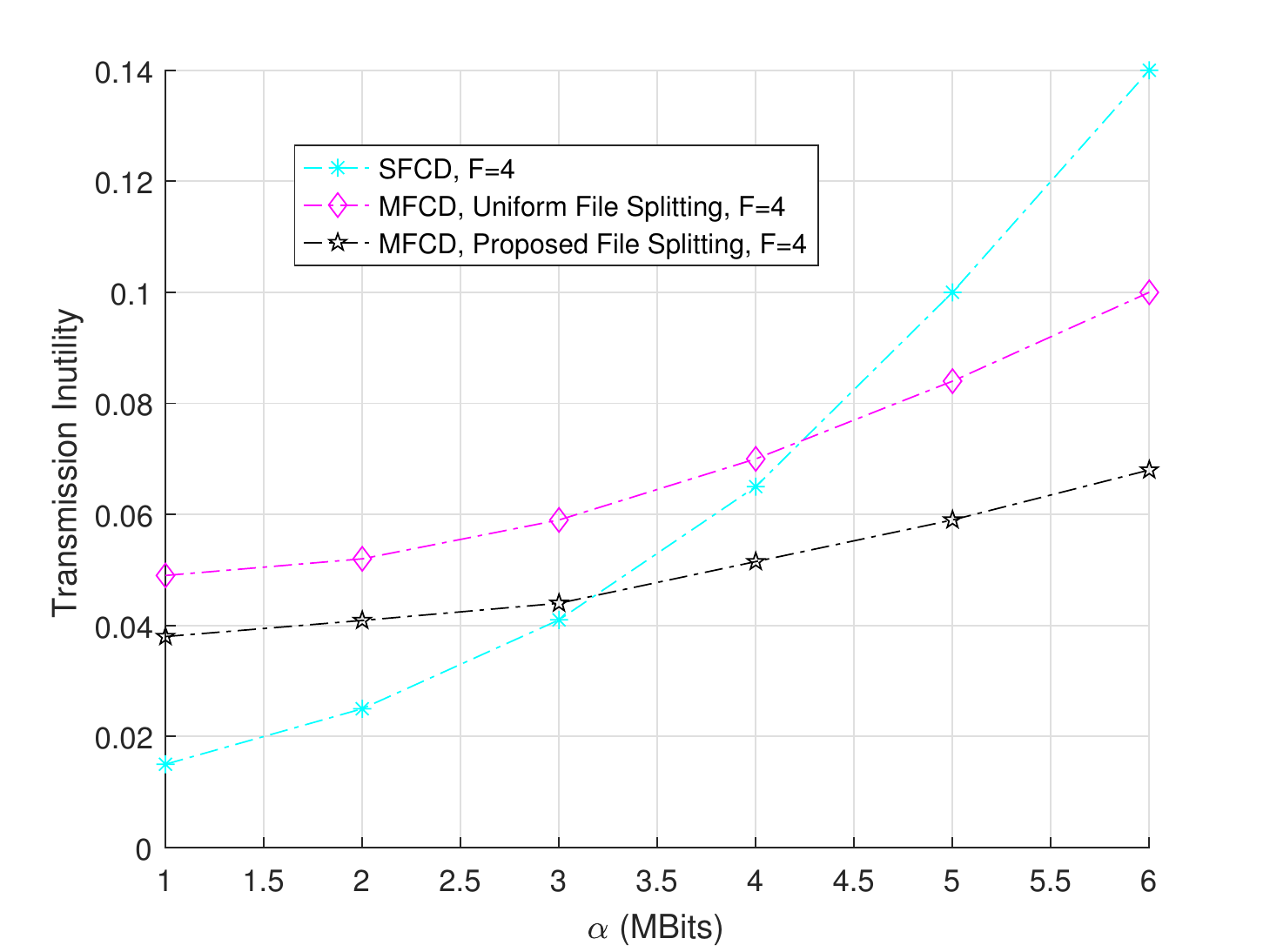}}
		\caption{(a) Energy efficiency, $EE$, \emph{vs}.
			the harvested energy per arrival, $\rho$.
			System parameters are: $B=2, M=28, F=2, U=2, N=8, Q=2, K=6, D=2, T=0.01$ s, $\forall k$, $V_b=10$ Mbits, $\forall b$, $\tilde{P}^{\text{Total}}=0.1$ Watt, $E^{\text{max}}_b=2$ Joule, $\forall b$, $\varepsilon_{h}=0.5$. (b) Outage probability, \emph{vs}.
			the number of frames, $F$.
			System parameters are: $B=2, M=28, U=4, N=8, Q=2, K=6, D=2, T=0.01$ s, $V_b=10$ Mbits, $\forall b$, $\tilde{P}^{\text{Total}}=1$ Watt, $E^{\text{max}}_b=2$ Joule, $\forall b$, $\rho=0.1$, Joule, $\varepsilon_{h}=0.5$.}
	\end{center}
\end{figure}

\subsection{Transmission Inutility}
In this section, we investigate the transmission inutility for the SFCD and MFCD schemes. The transmission inutility is defined by multiplying the outage probability in the transmission delay. The outage probability is defined as probability that there is not enough battery to send content files and the transmission delay is defined as number of frames to send files. Fig. \ref{Fig_Outage_probability_vs_F_diff_alpha}  demonstrates the transmission inutilities of our proposed schemes for different content file size. As can be seen, by increasing the size of file, the transmission inutility is increased for both schemes. This is due to the fact that there may be not enough harvested energy to send file, and the energy deficiency probability can be increased. Therefore, the outage probability approaches to one in sufficiently big size of files. This deterioration in the MFCD schemes are less than the SFCD scheme. Because in the MFCD schemes, the deficiency probability of energy can be reduced by dividing the content file into several parts and sending each part in different frames. In the proposed splitting scheme, we find the best fractional of content file for each frame which reduces the outage probability more than before. As shown in Fig. \ref{Fig_Outage_probability_vs_F_diff_alpha}, for larger content file sizes, the MFCD scheme has a higher efficiency in reducing the outage probability. As can be seen, for the size of content files less than 3 MBits, the SFCD scheme is better, while for the size of large files, the MFCD scheme is better.

\subsection{Effect of Channel Uncertainity}
Fig. \ref{Fig3} shows the access secrecy rate versus channel uncertainty for the SFCD scenario. We see
that at bigger channel uncertainty, the secrecy access rate clearly has low value. This is due to the fact that when the uncertainty increases, for the worst case scenario, we must guarantee the security for the  worst (biggest) channel value of eavesdroppers which leads to low values of secrecy rate. As can be seen, as the number of eavesdroppers increases, the secrecy access rate decreases due to the multiuser diversity gain for eavesdroppers. 

\begin{figure}[h]
	\begin{center}
		\subfigure[]{\label{Fig3}\includegraphics[width=3.in]
			{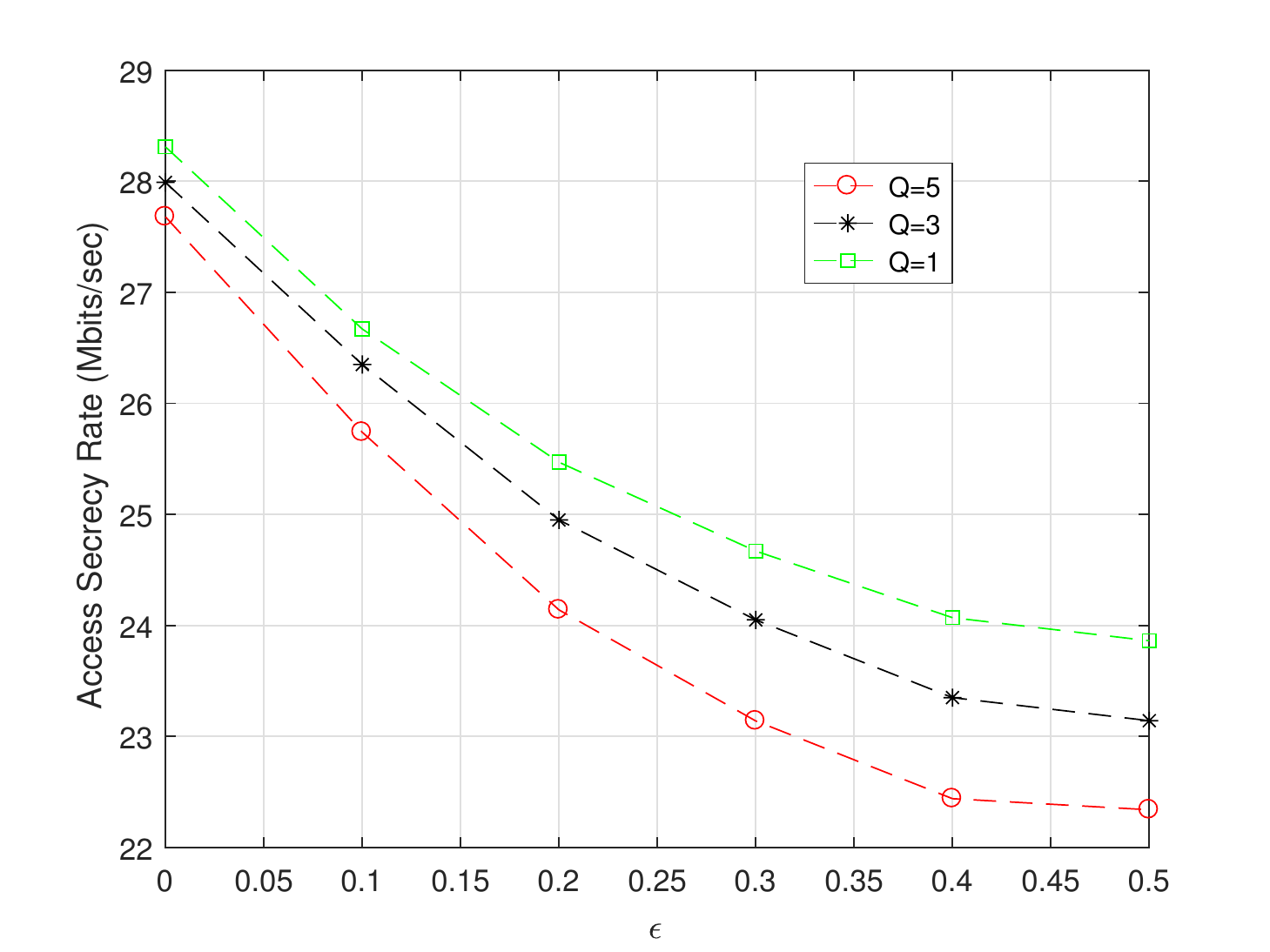}}
		\subfigure[]{\label{Fig_compare_joint_disjoint}\includegraphics[width=3.in]
			{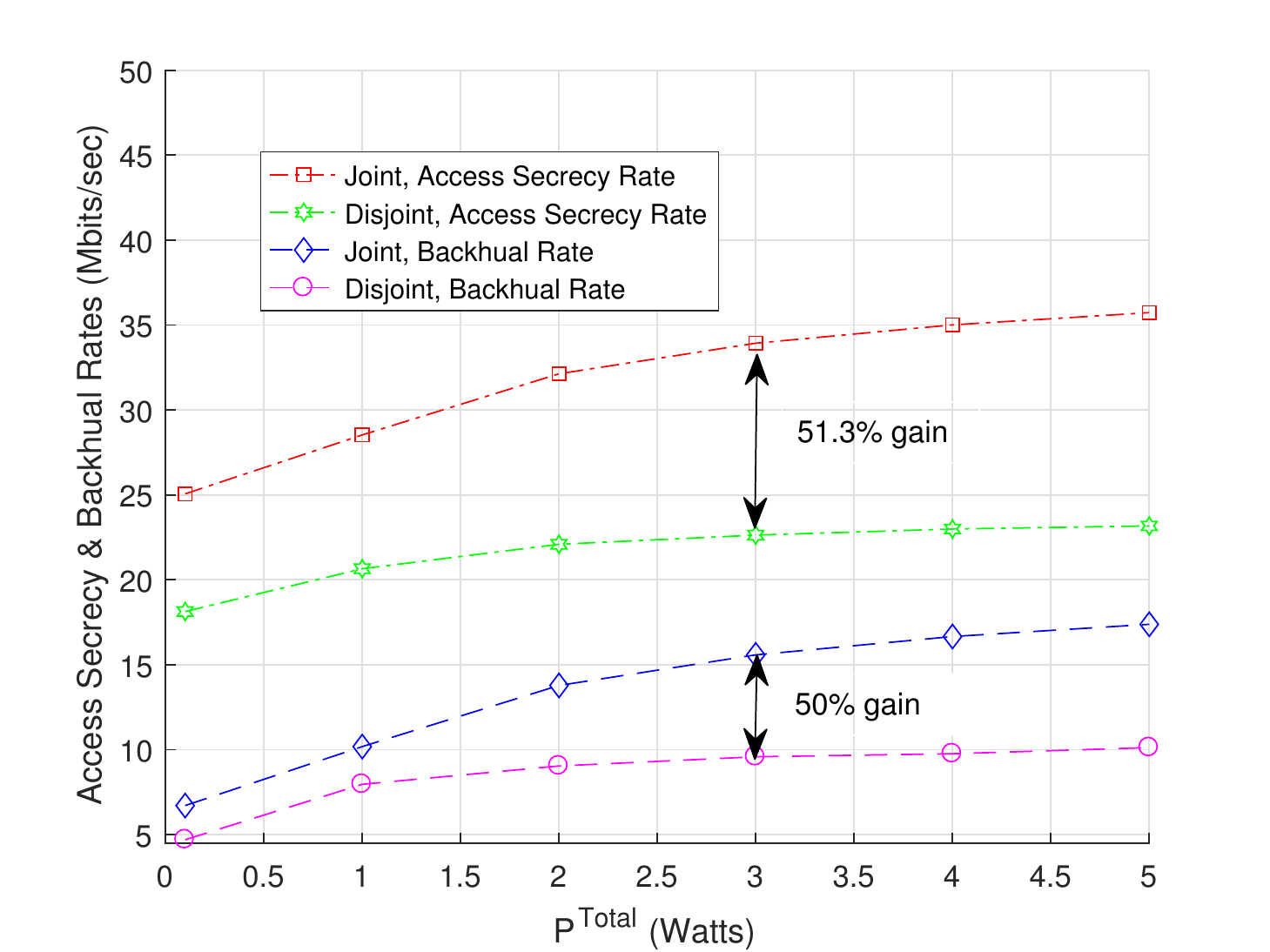}}
		\caption{(a) Access secrecy rate, $R^S$, \emph{vs}.
			the channel uncertainty, $\varepsilon_{h}$ for the SFCD scenario.
			System parameters are: $B=2, M=28, F=2, U=2, N=8, Q=1, K=6, D=2, T=0.01$ s, $\alpha_k=1$  Mbits, $\forall k$, $V_b=10$ Mbits, $\tilde{P}^{\text{Total}}=0.1$ Watts, $E^{\text{max}}_b=2$ Joule, $\forall b$, $\rho=0.5$ Joule. (b) Secrecy access and backhaul rates, $R^S$ and $\tilde{R}$, \emph{vs}.
			the backhaul transmission power, $P^{\text{Total}}$ for the SFCD scenario.
			System parameters are:$B=2, M=28, F=2, U=4, N=8, Q=1, K=6, D=2, T=0.01$
			s, $\alpha_k=1$ Mbits, $\forall k$, $V_b=10$ Mbits, $\forall b$, $E^{\text{max}}_b=5$ Joule, $\forall b$, $\rho=0.5$ Joule, $\varepsilon_{h}=0.5$.}
	\end{center}
\end{figure}

\subsection{Comparison Between Joint backhaul and access optimization and Disjoint Optimization Problem Solution}
Fig. \ref{Fig_compare_joint_disjoint} illustrates the comparison
between joint optimization and disjoint optimization problem
solutions versus different backhaul transmission power for the SFCD scenario. In solving the main problem disjointly, one must solve two optimization problems, one for the access part and the other for the backhaul part. Here, we first solve the backhaul problem since in most cases the backhaul capacity is the limiting factor. For the backhaul problem, the objective is maximization of the backhaul rate, and the constraint is the total transmit power of the backhaul. Solving this problem, the supported transmission rate of the backhaul is obtained which will be used in the access problem. The access problem is similar to the problem \eqref{eq--2}. The differences between the access problem and problem \eqref{eq--2} are that the optimization variables of the backhaul and the constraints relating to the backhaul are removed, and one additional constraint, which states the the access secrecy rate should be above the backhaul transmission rate (obtained in the backhaul problem), is included in the access problem.  It can be
noticed that the joint backhaul and access optimization approach has better solution than the disjoint backhaul and access
optimization approach. This is mainly because that in the joint scenario, the feasibility set of the optimization problem is bigger than the disjoint one. Indeed, in the access problem, we have the constraint which enforces that the access rate should above the backhaul transmission rate (obtained in the backhaul problem) which makes the feasibility set of the disjoint problems smaller than that of the joint one. From Fig. \ref{Fig_compare_joint_disjoint}, it can be observed that there
	exists nearly 51.3\% and 50\% performance gap between joint and disjoint approachs in terms of backhaul and access rates, respectively.

\subsection{Effect of super frame size}
Fig. \ref{Fig_Energy_efficiency_backhual_rate_access_rate_vs_F}
shows the variation of the EE with
the number of frames, $F$ for the SFCD scenario. It is seen that by increasing super frame size, the value of EE increases. In other words, by increasing super frame size, the transmitter can transmit data stream over different frames, then the secrecy access rate and EE will increase. For limited battery storage,  with increasing super frame size, at first the EE increases. However, with further increasing super frame size, due to energy overflow constraints, (\ref{eq-energy-3}), which enforce the transmitters to spend energy, the energy efficiency decreases. Note that, as the value of $\rho$ becomes larger, this decrease in EE happens in lower super frame sizes. For unlimited battery storage, since the overflow constrains, (\ref{eq-energy-3}), are absent, all the harvested energy is stored in the battery. In this case, increasing super frame size will increase the diversity gain, and hence, the energy efficiency increases.

\begin{figure}[h]
	\begin{center}
		\subfigure[]{\label{Fig_Energy_efficiency_backhual_rate_access_rate_vs_F}\includegraphics[width=3.in]
			{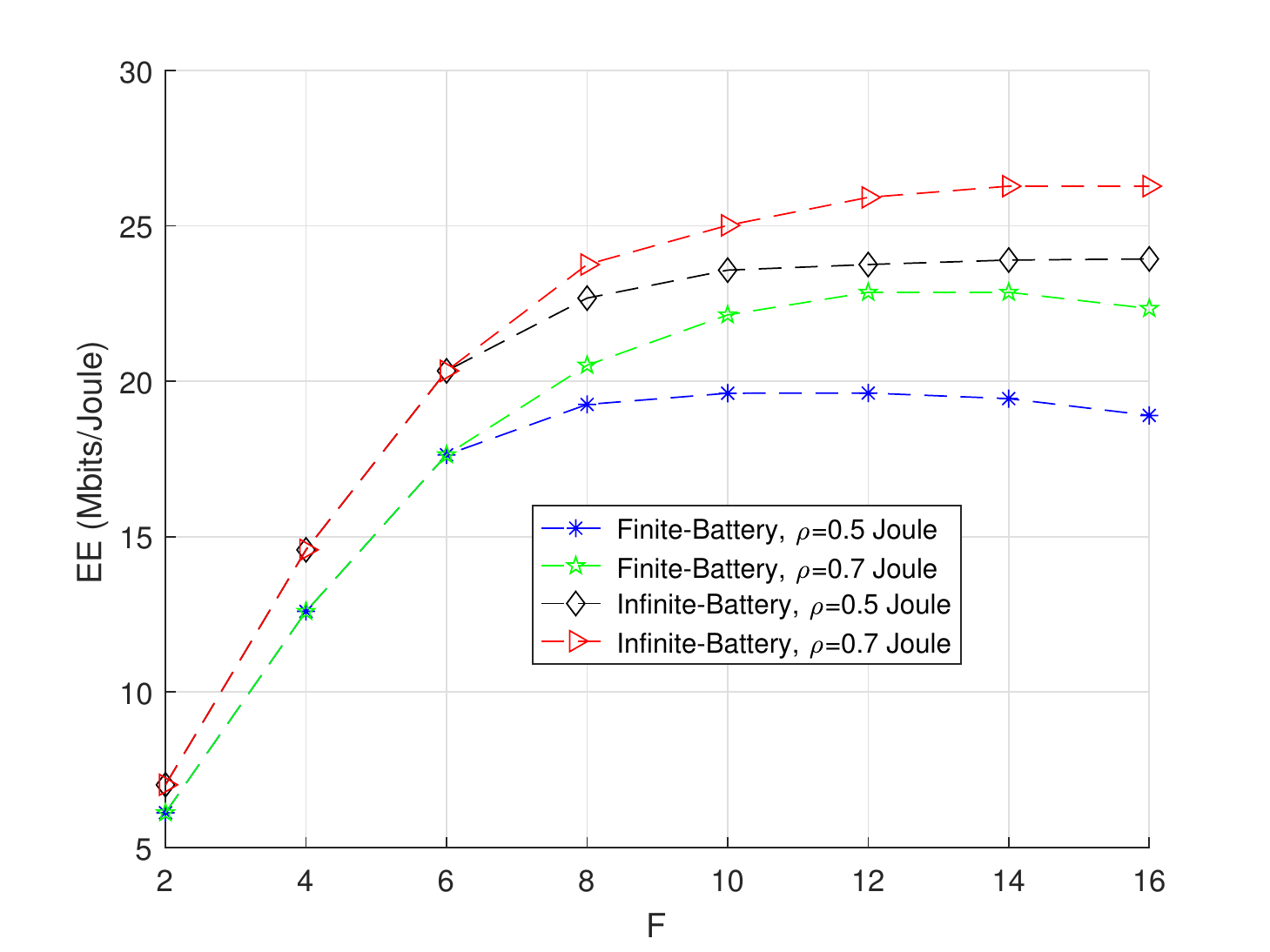}}
		\subfigure[]{\label{Fig_EE_vs_iteration_number}\includegraphics[width=3.in]
			{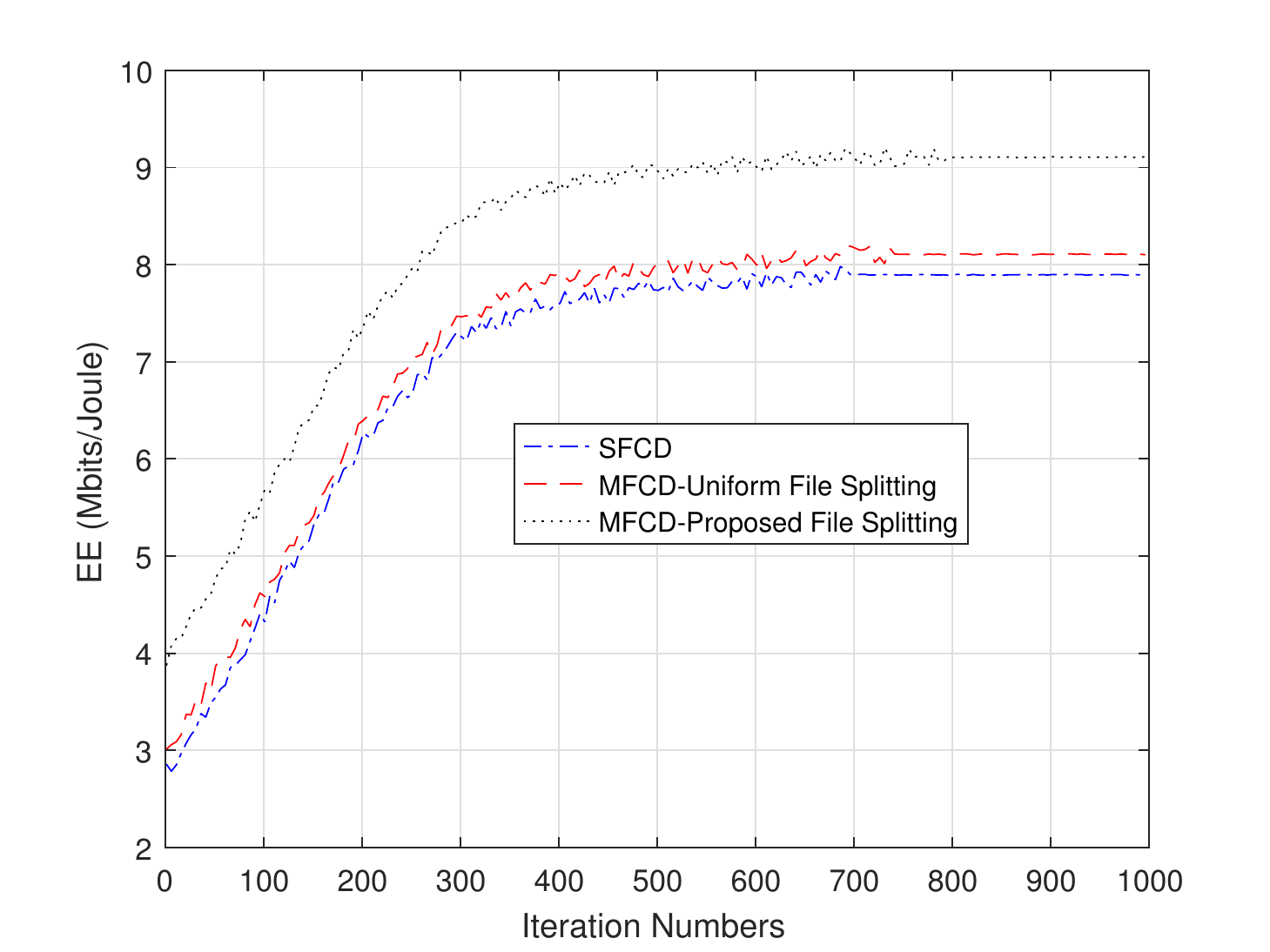}}
		\caption{(a) EE \emph{vs}. the super frame size, $F$.
			System parameters are: $B=2, M=28, U=4, N=8, Q=1, K=6, D=2, T=0.01$
			s, $\alpha_k=1$ Mbits, $\forall k$, $V_b=10$ Mbits, $\forall b$, $E^{\text{max}}_b=5$ Joule, $\forall b$, $\rho=0.5$ Joule, $\varepsilon_{h}=0.5$. (b) Energy efficiency, EE, \emph{vs}.
			the iteration number.
			System parameters are: $B=2, M=28, F=2, U=2, N=8, Q=1, K=6, D=2, T=0.01$ s, $\alpha_k=1$ Mbits, $\forall k$, $V_b=10$ Mbits, $\forall b$, $\tilde{P}^{\text{Total}}=0.1$ Watts, $\rho=0.1$ Joule, $E^{\text{max}}_b=2$ Joule, $\forall b$, $\varepsilon_{h}=0.5$.}
	\end{center}
\end{figure}

\subsection{The Convergence of the Proposed Algorithm}
In this part, we investigate the performance of the proposed resource allocation
algorithm. In Fig.
\ref{Fig_EE_vs_iteration_number}, we show EE after
each iteration at the proposed alternate optimization algorithm. As can be seen, the convergence of the proposed
algorithm can averagely be achieved within 700 iterations.

\section{Conclusion}\label{Conclusion}
In this paper, we provided a unified framework for radio resources allocation and
content placement considering the physical layer security and the channel uncertainty to provide higher energy efficiency. To do so,
we considered downlink SCMA scenarios, and we aimed at maximizing the
worst case energy efficiency subject to system constraints which determines the radio resources allocation and content placement
parameter. Moreover, we proposed two novel content delivery scenarios: 1) single frame content delivery, and 2) multiple frames content delivery. In the first scenario, the requested content by each user is served over one frame. However, in the second scenario, the requested content by each user can be delivered over several frames. Since the optimization problems are noncovex and NP-hard, we provided an iterative method converging to a local solution. Finally, we showed the
resulting secrecy access rate, backhaul rate, and energy efficiency for
different values of maximum backhaul transmit power as well as different number of users and various content size. In addition, we compared the performance of the proposed caching scheme with the existing  traditional caching schemes. Based on simulation results, via our proposed caching scheme, the performance is approximately improved by 14\% and 21\% compared to the most popular and random caching schemes, respectively. Moreover, it can be seen that the MFCD scheme can approximately enhance the system performance by 5.2\% and 11.1\% for small and large files, respectively.

\begin{appendices}
\section{Proof of Lemma 1}\label{appendix-A}
\ We jointly find the optimization variables $\mathbf{p},\tilde{\mathbf{p}},\mathbf{s},\boldsymbol{\theta}$, and $\boldsymbol{\zeta}$ such that the EE of proposed system is maximized. Hence, (\ref{eq--2}) is MINLP and non-convex. We assume that the optimization variables $\tilde{\mathbf{p}},\boldsymbol{\theta}$, and $\boldsymbol{\zeta}$ are constant. In access link, we also assume that one subcarrier is exclusively assigned to at most $D$ users within the cell. We consider the downlink of an SCMA-based access link consisting of $N$ subcarriers, $U$ users and $D=2$. By assuming that the special $u^{\text{th}}$ user's channel gain on all subcarriers is the largest among all users, the optimal power assigned to user
$u$ is equal to $p^{mt}_{bu}/N$, where $p^{mt}_{bu}$
is the transmit power assigned to user $u$ at BS $b$ at frame $t$ on codebook $m$. Then, the challenge
is how to allocate the remaining power resource $E^t_b/T-p^{mt}_{bu}$ to $U-1$ users over all subcarriers.
Thus, no subcarrier can be assigned to more than one user. Therefore, a special case
of power and codebook optimization problem with $D>1$ is equivalent to the NP-hard problem considered in \cite{lei2015joint}, and the result follows. Finally, it can be concluded that the main problem (\ref{eq--2}) is also NP-hard.
	
\section{Proof of Theorem 1}\label{appendix-B}
\ In accordance to the foregoing discussions, for (\ref{AS-1}) with
a given ${\mathbf{e}_h}_{\varrho}$,
$(\mathbf{p}_{\varrho+1},\tilde{\mathbf{p}}_{\varrho+1},\mathbf{s}_{\varrho+1},\boldsymbol{\theta}_{\varrho+1}
,\boldsymbol{\zeta}_{\varrho+1})$ is its optimal solution, while
$(\mathbf{p}_{\varrho},\tilde{\mathbf{p}}_{\varrho},\mathbf{s}_{\varrho},\boldsymbol{\theta}_{\varrho}
,\boldsymbol{\zeta}_{\varrho})$ is only its feasible solution. We
get that
\begin{align}
&\Xi_{EE}(\mathbf{p}_{\varrho+1},\tilde{\mathbf{p}}_{\varrho+1},\mathbf{s}_{\varrho+1},
\boldsymbol{\theta}_{\varrho+1}
,\boldsymbol{\zeta}_{\varrho+1},{\mathbf{e}_h}_{\varrho})\leq\\\nonumber&
\Xi_{EE}(\mathbf{p}_{\varrho},\tilde{\mathbf{p}}_{\varrho},\mathbf{s}_{\varrho},\boldsymbol{\theta}_{\varrho}
,\boldsymbol{\zeta}_{\varrho},{\mathbf{e}_h}_{\varrho}).
\end{align}
Likewise, for (\ref{AS-2}) with a given $\mathbf{p}_{\varrho}$,
$(\tilde{\mathbf{p}}_{\varrho+1},\mathbf{s}_{\varrho+1},\boldsymbol{\theta}_{\varrho+1}
,\boldsymbol{\zeta}_{\varrho+1},{\mathbf{e}_h}_{\varrho+1})$ is
its optimal solution, while
$(\tilde{\mathbf{p}}_{\varrho},\mathbf{s}_{\varrho},\boldsymbol{\theta}_{\varrho}
,\boldsymbol{\zeta}_{\varrho},{\mathbf{e}_h}_{\varrho})$ is only
its feasible solution. It follows that
\begin{align}
&\Xi_{EE}(\mathbf{p}_{\varrho},\tilde{\mathbf{p}}_{\varrho+1},\mathbf{s}_{\varrho+1},
\boldsymbol{\theta}_{\varrho+1}
,\boldsymbol{\zeta}_{\varrho+1},{\mathbf{e}_h}_{\varrho+1})\leq\\\nonumber&
\Xi_{EE}(\mathbf{p}_{\varrho},\tilde{\mathbf{p}}_{\varrho},\mathbf{s}_{\varrho},\boldsymbol{\theta}_{\varrho}
,\boldsymbol{\zeta}_{\varrho},{\mathbf{e}_h}_{\varrho}).
\end{align}
For relations (\ref{AS-3}), (\ref{AS-4}), and
(\ref{AS-6}), this trend is similar. It is naturally concluded that
\begin{align}
&\Xi_{EE}(\mathbf{p}_{\varrho+1},\tilde{\mathbf{p}}_{\varrho+1},\mathbf{s}_{\varrho+1},\boldsymbol{\theta}_{\varrho+1}
,\boldsymbol{\zeta}_{\varrho+1},{\mathbf{e}_h}_{\varrho+1})\leq\\\nonumber&
\Xi_{EE}(\mathbf{p}_{\varrho},\tilde{\mathbf{p}}_{\varrho},\mathbf{s}_{\varrho},\boldsymbol{\theta}_{\varrho}
,\boldsymbol{\zeta}_{\varrho},{\mathbf{e}_h}_{\varrho}).
\end{align}

\section{Proof of Theorem 2}\label{appendix-C}
\ Because of the convexity of $R^{\text{E},mt}_{buq1}$, it follows
that
\begin{align}\label{eq-001}
&R^{\text{E},mt}_{buq1}({\varrho+1})\geq\\\nonumber&
R^{\text{E},mt}_{buq1}({\varrho})-\left\langle\nabla
R^{\text{E},t}_{bu1}({\varrho}),
p^{mt}_{bu}({\varrho+1})-p^{mt}_{bu}(\varrho)\right\rangle,
\end{align}
for $p^{mt}_{bu}(\varrho)$ and $p^{mt}_{bu}(\varrho+1)$ in the
feasible domain. We can deduce that
\begin{align}\label{eq-002}
&-\left\langle\nabla
R^{\text{E},t}_{bu1}({\varrho}),
p^{mt}_{bu}({\varrho+1})-p^{mt}_{bu}(\varrho)\right\rangle \\\nonumber&-(R^{\text{E},mt}_{buq2}({\varrho+1})-R^{\text{E},mt}_{buq1}({\varrho}))\leq
-(R^{\text{E},mt}_{buq2}({\varrho})-R^{\text{E},mt}_{buq1}({\varrho})),
\end{align}

Combined with (\ref{eq-001}) and (\ref{eq-002}), we conclude that
\begin{equation}\label{eq-003}
-(R^{\text{E},mt}_{buq2}({\varrho+1})-R^{\text{E},mt}_{buq1}({\varrho+1}))
\leq
-(R^{\text{E},mt}_{buq2}({\varrho})-R^{\text{E},mt}_{buq1}({\varrho})).
\end{equation}

Obviously, the current value $R^{\text{E},mt}_{buq}({\varrho+1})$
is smaller than the previous value
$R^{\text{E},mt}_{buq}({\varrho})$ while the current solution
$p^{mt}_{bu}(\varrho+1)$ is better than the previous solution
$p^{mt}_{bu}(\varrho)$. As a result, the theorem is proved.
\end{appendices}

 \hyphenation{op-tical net-works semi-conduc-tor}
\bibliographystyle{IEEEtran}
\bibliography{IEEEabrv,Bibliography}
\end{document}